\documentclass[11pt, a4paper, oneside,reqno]{amsart}

\usepackage{style}
\usepackage[square,numbers]{natbib}
\usepackage{amsmath}
\usepackage{amssymb}
\usepackage{mathtools}
\usepackage{amsthm}
\usepackage{bm}
\usepackage{comment}
\usepackage{booktabs}
\usepackage{microtype}
\usepackage{tikz}
\usepackage{pgfplots}
\pgfplotsset{compat=newest}
\pgfplotsset{compat=1.18}
\usepgfplotslibrary{fillbetween}

\usepackage[most]{tcolorbox}

\tcolorboxenvironment{example}{
    colback=blue!8,
    colframe=blue!50,
    boxrule=0.5pt,
    arc=2pt,
    left=6pt,
    right=6pt,
    top=6pt,
    bottom=6pt,
}

\usepackage[textsize=tiny]{todonotes}

\usepackage[linesnumbered,ruled]{algorithm2e}
\usepackage{graphicx}
\usepackage{eurosym}
\usepackage{makecell}
\usepackage{subcaption}
\newtheorem{theorem}{Theorem}[section]
\newtheorem{assumption}[theorem]{Assumption}
\newtheorem{proposition}[theorem]{Proposition}

\newtheorem{lemma}{Lemma}
\newtheorem{example}[theorem]{Example}

\newtheorem*{claim*}{Claim}
\newtheorem*{problem}{Problem}

\graphicspath{{Fig/}}

\allowdisplaybreaks

\title[]{Adaptive Incentive Design in Dynamic \\ Principal-Agent Problem via Kernelized Bandits}

\author[]{Arghya Mallick, Anuj S. Vora, Sergio Grammatico, and Peyman Mohajerin Esfahani}
\thanks{This work was partly funded by EU project DriVe2X (grant no. 101056934), and the ERC project TRUST (grant no. 949796). Arghya Mallick, and Sergio Grammatico are with Delft Center for Systems and Control, Delft University of Technology, 
Delft, The Netherlands. Email: \texttt{\{A.Mallick,\,  s.grammatico\}}@tudelft.nl; Anuj S. Vora is with Jio Institute, Navi Mumbai, India. Email: \texttt{anuj.vora}@jioinstitute.edu.in; Peyman Mohajerin Esfahani is with the University of Toronto, Toronto, Canada. Email: \texttt{P.MohajerinEsfahani}@utoronto.ca}
\thanks{} 

\begin{document}

\begin{abstract}
We consider the dynamic principal-agent problem under asymmetric information, wherein a principal sequentially designs contracts to incentivize an agent with unknown preferences and hidden actions. A fundamental bottleneck in the existing literature is the assumption of deterministic agent utility, which renders the principal's expected utility discontinuous and forces computationally intractable discretizations of the contract space. In this paper, we address this limitation by introducing a stochastic counterpart into the agent's utility model, capturing the inherent physical and behavioral variations in realistic subsystems. We formally prove that this stochastic formulation restores the continuity of the principal's expected utility. Leveraging this continuous geometric structure, we formulate the interaction as a structured multi-armed bandit problem subject to heteroscedastic noise. We propose a \texttt{Heteroscedastic GP-UCB} algorithm that utilizes a Neural Network (Arcsin) kernel, chosen to capture the non-stationary, sigmoidal geometry of the utility landscape. For an $m$-dimensional compact contract space, we establish a high-probability cumulative regret bound of $\mathcal{O}\left(\sqrt{T}(\log T)^{m+1}\right)$. Finally, we demonstrate the practical efficacy of our theoretical framework by formulating the Vehicle-to-Grid (V2G) incentive design problem, proving its equivalence to a dynamic principal-agent problem, and showing superior economic performance for grid aggregators.
\end{abstract}
\maketitle
\section{Introduction}
\label{sec:introduction}
We study the classical principal-agent problem \cite{laffont2002theory} in a dynamic setting. The broader significance of this problem—underscored by the 2016 Nobel Memorial Prize in Economic Sciences \cite{wang2023deep}—lies in its ability to model asymmetric information. Specifically, a principal sequentially designs contracts to incentivize a strategic agent whose actions are hidden (unobservable to the principal) but stochastically govern the task outcome and, consequently, the principal's payoff. Within the control community, such incentive design has historically been investigated through the lens of Stackelberg games \cite{ho1981information,ho1984interactions,lauffer2023no}, addressing applications ranging from demand response \cite{zhou2017incentive} to network congestion \cite{barrera2014dynamic}. While recent works have made rich theoretical strides in adaptive incentive design \cite{ratliff2020adaptive,maheshwari2025adaptive,lauffer2023no}, they frequently rely on restrictive assumptions that deviate from the classical setting. For instance, the approaches in \cite{ratliff2020adaptive} and \cite{lauffer2023no} propose no-regret learning algorithms for dynamic Stackelberg games, but fundamentally assume that the leader (principal) can directly observe the follower's (agent's) actions at each round, or that the follower's utility is linearly parameterized.

In contrast to these idealized assumptions, we preserve the structural integrity of the original principal-agent problem, specifically the core challenge of the agent's unobservable actions (known as \textit{moral hazard}). To navigate this, we cast the dynamic contract design problem as a structured multi-armed bandit problem \cite{mersereau2009structured,cope2009regret}. Bandit frameworks naturally accommodate sequential decision-making under partial feedback among the principal, the agent, and the environment \cite{asawa2002multi,zhu2023distributed,yang2025stochastically,liu2024thompson}. Recently, bandit-based algorithms have made foundational contributions to automated contract design on modern crowdsourcing platforms, such as \textit{Amazon Mechanical Turk} \cite{ho2014adaptive,harris2011you}. Beyond digital markets, this adaptive framework holds immense potential for decentralized cyber-physical systems, most notably Vehicle-to-Grid (V2G) technology \cite{liu2013opportunities}. V2G enables electric vehicles (EVs) to act as distributed energy storage, providing critical balancing services to the power grid. However, while the physical infrastructure is increasingly mature \cite{brooks2002vehicle}, widespread adoption is severely bottlenecked by misaligned economic incentives and unobservable user costs \cite{van2021factors}. Motivated by these research gaps, we focus on developing computationally efficient algorithms while modeling a principal-agent-compatible incentive framework for the end users of `V2G' technology.

Formally, in each round $t$, the principal posts a contract $x_t \in \mathcal{X}:= [0,1]^m$, that specifies an incentive payment $(x_t)_i$ for the $i$-th outcome of the task, with $i\in [m]:= \{1, \dots m\}$. An agent observes the contract, takes an action $a^{\star}_t \in \mathcal{A} $, and the outcome of the task is realized. The uncoupled utility (or payoff) function of the principal is $U^{\mathrm{P}}: \mathcal{X} \times \mathcal{A} \rightarrow \mathbb{R}$ and the utility function of the agent is $U^{\mathrm{A}}: \mathcal{X} \times \mathcal{A}\rightarrow \mathbb{R}$. We assume the agent is rational and always takes an action that maximizes its utility $U^{\mathrm{A}}$. For simplicity, dropping the time $t$ as a subscript, we compactly formulate the problem~as 
\begin{subequations}\label{agent_bp}
\begin{align}
 a^\star(x)
&:= \arg\max_{a \in \mathcal{A}}  U^{\mathrm{A}}(x,a),
\label{agent_br} \\
 x^\star
&:= \arg\max_{x \in \mathcal{X}} \; \Big\{U^{\mathrm{PA}}(x) := U^{\mathrm{P}}\!\bigl(x,a^{\star}(x)\bigr)\Big\}, \label{opt_contract_eq}
\end{align}
\end{subequations}
where $a^\star(x)$ is the best response of the agent against the contract $x$, $U^{\mathrm{PA}}: \mathcal{X}\rightarrow \mathbb{R}$ is the principal's utility with embedded best response of the agent, and $x^\star$ denotes the optimal contract. The main technical challenge is that the principal does not explicitly know $U^{\mathrm{PA}}$ because of the agent's hidden action $a^\star$. Our goal is to learn $x^\star$ via online random feedback from the principal's utility $U^{\mathrm{PA}}$. To quantify the sub-optimality of the principal's strategy in this sequential interaction, we use the traditional notion of cumulative regret as
\begin{align} \label{reg_def}
    R_T := T \cdot U^{\mathrm{PA}}(x^\star) - \sum_{t=1}^T U^{\mathrm{PA}}(x_t),
\end{align}
where $x^\star$ is as defined in \eqref{opt_contract_eq}. We are interested in designing an algorithm to propose a sequence of $\{x_t\}_{t \le T}$ using the available feedback information from the environment so that we ensure the upper bound of $R_T$ grows sub-linearly in $T$.
\subsection{Related literature}
\textit{Computational Contract Design.}
While the principal-agent problem is rooted in economic theory \cite{holmstrom1979moral}, recent literature has focused on its computational aspects. Authors in \cite{babaioff2006combinatorial} establish that computing optimal contracts in combinatorial settings is NP-hard in general. Hence, research shifted towards identifying robust, simple contract structures, such as linear contracts~\cite{dutting2019simple}, which provide approximation guarantees in static environments. However, these works assume that the principal has full knowledge of the agent's type and payoff function, an assumption that rarely holds in~practice.

\textit{Learning Contracts with Bandit Feedback.}
In the repeated principal-agent setting, the problem is modeled as a multi-armed bandit with a finite number of arms. Authors in \cite{ho2014adaptive} pioneer this direction with the \texttt{Agnostic Zooming} algorithm. Their approach relies on adaptively discretizing the contract space; however, the analysis necessitates strong regularity assumptions, specifically that contracts are monotone and outcomes satisfy first-order stochastic dominance. Authors in \cite{zhu2023sample} relaxed these conditions, offering guarantees for general utility functions. More recently, authors in \cite{bacchiocchi2023learning} proposed the \texttt{Discover and Cover} algorithm to address the problem with finite agent actions. While this approach avoids the geometric assumptions of earlier works, it suffers from the \textit{curse of dimensionality}, as the complexity scales exponentially with the number of agent actions.

\textit{Continuity of principal's utility.} The principal's expected utility $U^{\mathrm{PA}}$ is inherently discontinuous, formally characterized by \cite{wang2023deep} as a piecewise-affine function with sharp transitions. Prior dynamic algorithms \cite{ho2014adaptive, zhu2023sample} treat this discontinuity as an obstacle that requires inefficient contract-space discretization. In contrast, our work introduces a stochastic utility framework that smooths out these discontinuities, enabling the application of kernelized bandits \cite{chowdhury2017kernelized} directly in the continuous domain.
\begin{figure}
    \centering
    \includegraphics[width=0.7\linewidth]{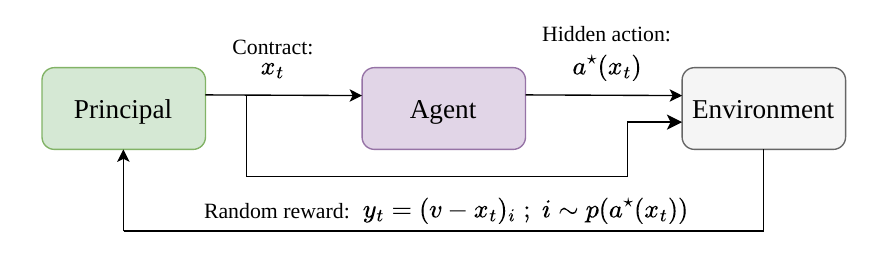}
    \caption{Schematic of principal-agent feedback model \eqref{agent_bp}. }
    \label{pa_fig}
\end{figure}

\subsection{Contributions}
\begin{enumerate}
\renewcommand{\labelenumi}{(\roman{enumi})}
\item \textbf{Modeling: Resolving discontinuity of $U^{\mathrm{PA}}$.} We identify that the discontinuity of the principal's utility $U^{\mathrm{PA}}$, a persistent barrier in dynamic contract design, stems from the idealized assumption of deterministic agent utilities. By incorporating a stochastic counterpart into the agent's utility function (capturing unmodeled physical transients and behavioral fluctuations), we prove that the principal's expected utility $\Tilde{U}^{\mathrm{PA}}$ becomes continuous over its domain (Theorem~\ref{main_lemma}). This modeling shift provides a more realistic representation of the problem and fundamentally resolves the need for the intractable discretization relied upon by existing state-of-the-art methods \cite{bacchiocchi2023learning,ho2014adaptive,zhu2023sample}.
\item \textbf{Algorithm: Sub-linear regret via tailored kernel.} Leveraging the continuity result of $\Tilde{U}^{\mathrm{PA}}$, we propose the \texttt{Heteroscedastic GP-UCB} algorithm paired with a Neural Network kernel, specifically chosen to capture the non-stationary, sigmoidal geometry of $\Tilde{U}^{\mathrm{PA}}$. By bounding the eigenvalue tail sum of this specific kernel (Proposition~\ref{eigen_val_lemma}), we establish a formal high-probability regret bound of: $\mathbb{P}\left[ \tilde{R}_T \leq \mathcal{O}\left( \sqrt{T}(\log T)^{m+1} \right ) \right] \geq 1-\delta$ (Theorem~\ref{main_theorem}), where $m$ denotes the dimension of contract space $\mathcal{X}$, and $\delta>0$ is the confidence level.
\item \textbf{Application: V2G incentive design.} Finally, we demonstrate the practical relevance of our framework by addressing the Vehicle-to-Grid (V2G) incentive design problem. We formulate a smart-charging scheme where aggregators incentivize electric vehicle owners to provide ancillary services. We prove (Proposition~\ref{prop_1}) that this complex engineering problem, governed by stochastic battery dynamics and user preferences, formally reduces to a dynamic principal-agent formulation.
\end{enumerate}
It is noteworthy to mention that our proposed algorithm overcomes the best-known regret bound of the prior work \cite{bacchiocchi2023learning} (i.e., $R_T \leq\mathcal{O}(m^nT^{4/5} \log T)$) by improving the original principal-agent problem formulation (i.e., from $U^{\mathrm{PA}}$ to $\Tilde{U}^{\mathrm{PA}}$). We do not claim any improvement on $R_T$ for the original (deterministic) problem formulation. Moreover, while $R_T$ of \cite{bacchiocchi2023learning} scales exponentially with the number of actions of the agent (i.e., $m^n$), the regret bound of our work (i.e., $\Tilde{R}_T$) entirely eliminates this exponential dependence on $n$, where $n$ and $m$ denote the number of actions of the agent and the number of outcomes, respectively. In a nutshell, our methodology establishes that adopting a more realistic system model enables a faster learning rate. 


\section{Principal-Agent Feedback Model} 
\label{sec:problem}
Building on the introduction, we now discuss the principal-agent model's formulation in detail. After the agent executes the best response following \eqref{agent_br}, an outcome $i \in [m]$ of the task is realized. The value for any outcome $i \in [m]$ for the principal is denoted by $(v)_i \in[0,1]$ and is known to the principal. Moreover, the outcome of a task is realized following a conditional probability distribution defined by the production function $p: \mathcal{A} \rightarrow \Delta_m$, where $\Delta_m$ is the probability simplex over $m$ outcomes. We define agents' and principals' utilities, respectively, 
\begin{subequations}\label{agent_utility}
\begin{align}
& U^{\mathrm{A}}(x,a)= \sum_{i\in [m]} (x)_i\;p_i(a)-c(a), 
\label{agent_utility_first} \\
&U^{\mathrm{P}}(x,a)= \sum_{i\in [m]}(v-x)_i \;p_i(a), \label{u^p_fcn}
\end{align}
\end{subequations}
where $p_i(a)$ is the probability of the $i$-th outcome conditional on action $a$, and $c: \mathcal{A}\rightarrow \mathbb{R}$ denotes the cost function of taking actions by the agent. As stated previously, $a^\star(x)$ is not observable to the principal; however, the principal observes the following random reward conditioned on $a^\star(x)$:
\begin{align}\label{reward_eq}
    y= (v-x)_i,\; \text{with the random index's law}\; i\sim p(a^\star(x)).
\end{align}
A schematic of the principal-agent model described in \eqref{agent_bp}-\eqref{reward_eq} is shown in Figure~\ref{pa_fig}. We now discuss state-of-the-art methods for addressing the principal-agent problem. Furthermore, we analyze the key technical challenges faced by existing solution approaches. 
\subsection{State-of-the-art results} 
Authors in \cite{ho2014adaptive} propose an adaptive algorithm called \texttt{Agnostic Zooming} for this problem. Their algorithm provides the following upper bound on regret: $R_T (X_{\text{cand}}) \leq \mathcal{O}(\alpha \log T)T^{\frac{m+1}{m+2}}$, where $X_{\text{cand}} \subset [0,1]^m$ represents the set of monotone contracts. In addition, they also assume \textit{first order stochastic dominance} of contracts, which makes their regret bound to be less attractive. 

Authors in \cite{zhu2023sample} recently proposed an algorithm where they lifted the monotonicity and \textit{first order stochastic dominance} assumptions of \cite{ho2014adaptive}. They provide the upper bound of the regret in expectation as $R_T \leq \mathcal{O}(\sqrt{m}\log(T) T^{\frac{2m}{2m+1}})$. However, their proposed algorithm cannot be implemented in reality. In fact, one of the steps in their algorithm requires $\mathcal{X}$ to be discretized using certain methods from spherical coding in information theory, which remains an open problem to date (Section $3.1$ of \cite{zhu2023sample}).

Recently, Authors in \cite{bacchiocchi2023learning} propose an algorithm considering $\mathcal{A}$ is finite. They provide a probabilistic regret bound $\mathbb{P}\left [ R_T \leq \mathcal{O}(m^n \mathcal{I}\log (T)T^{4/5}) \right ]\geq 1-\delta,$ where $\mathcal{I}$ depends on $m, n$ polynomially, and $n$ is the number of actions available to agent. Here, the constants deteriorate performance significantly as $m$ and $n$ increase. This algorithm requires the private parameter~$n$. Apart from the main results above, the survey \cite {dutting2024algorithmic} presents other weaker results on the computational aspects of contracts. 
\subsection{Main technical challenges} \label{sub_sec_discont.}
Assuming $\mathcal{A}$ is finite, the authors in \cite{wang2023deep} establish that the principal's utility function $U^{\mathrm{PA}}$ in \eqref{opt_contract_eq} is a piecewise affine function; it can be discontinuous on the boundary of linear pieces, and the global optimal contract is on the boundary of a linear piece. Let us first understand why $U^{\mathrm{PA}}$ is discontinuous on its domain with the help of a simple example. Our understanding of the root cause of this technical challenge will eventually help us to resolve it in later sections. 

Assuming $\mathcal{A}$ is finite, the utility for an agent is $U^{\mathrm{A}}(x,a) = \langle x,p(a) \rangle - c(a),$ where $a \in \mathcal{A} := \{a_1, a_2, \dots, a_n\}$. For a fixed contract $x$, the objective function $U^{\mathrm{A}}$ admits multiple optimizers only if $x$ lies on an indifference hyperplane between distinct actions $a_i$ and $a_j$, defined as $H_{i,j} := \left\{x \in \mathcal{X} \mid U^{\mathrm{A}}(x, a_i) = U^{\mathrm{A}}(x, a_j) \right\}.$ Consider a contract $x \in H_{i,j}$ such that the optimal action set is restricted to the pair $\{a_i, a_j\}$, meaning $U^{\mathrm{A}}(x, a_i) = U^{\mathrm{A}}(x, a_j) > U^{\mathrm{A}}(x, a_k)$ for all $k \notin \{i, j\}$. If $p(a_i) \neq p(a_j)$, then the principal's utility $U^{\mathrm{PA}}$ exhibits a discontinuity across the hyperplane $H_{i,j}$. Letting $h_{i,j}(x) := U^{\mathrm{A}}(x, a_i) - U^{\mathrm{A}}(x, a_j)$ denote the affine function defining the boundary, the function $U^{\mathrm{PA}}$ takes the following piecewise form in the local neighborhood of $x$:
\begin{equation}
\begin{aligned}
U^{\mathrm{PA}}(x)
&=  \langle (v-x),
   p(a^\star(x))\rangle + 
\begin{cases}
\langle (v-x),p(a_j) \rangle, & \text{if } h_{i,j}(x) < 0, \\[2pt]
\langle (v-x), p(a_i) \rangle, & \text{if } h_{i,j}(x) > 0.
\end{cases}
\end{aligned}
\end{equation}
Note that $\mathcal{X}$ is partitioned by at most $\frac{n(n-1)}{2}$ such discontinuous hyperplanes, corresponding to $n$ actions of the agent.
Let us now construct an instance of the so-called `\textit{High-low example}' \cite{ho2014adaptive} to empirically visualize the landscape of $U^{\mathrm{PA}}$. 
\begin{example}[High-low case study] \label{high_low_example}
The principal posts a contract $x \in [0,1]^2$ with two outcomes (low and high), and the agent takes action $a \in \{a^{\text{l}}(\text{low effort}), $ $ a^{\text{h}}(\text{high effort})\}$ with corresponding private costs $c(a^{\text{l}}),c(a^{\text{h}})$, respectively. Low and high outcomes bring respective values $v=[v^{\text{l}},v^{\text{h}}]^{\top}$ to the principal. To empirically construct $U^{\mathrm{PA}}$, we consider $p(a^{\text{l}})=[0.9,0.1],\:p(a^{\text{h}})=[0.1,0.9], \: c(a^{\text{l}})=0, \: c(a^{\text{h}})=0.3$, and $v=[0.1,0.8]^{\top}$. Then we discretize $[0,1]^2$ with $(1000 \times 1000)$ data points and sample each data point $1000$ times to empirically approximate $U^{\mathrm{PA}}$. Figure~\ref{jump_fig} illustrates the geometry of the empirically constructed ${U}^\mathrm{PA}$. The $x$-axis represents contract $(x)_1$ for low outcome, and the $y$-axis denotes $(x)_2$ for high outcome. Unsurprisingly, we find it to be discontinuous between two affine pieces.
\end{example}

\begin{figure*}
 	\begin{subfigure}{0.3\linewidth}
    		\centering
    		\includegraphics[width=1.3\linewidth]{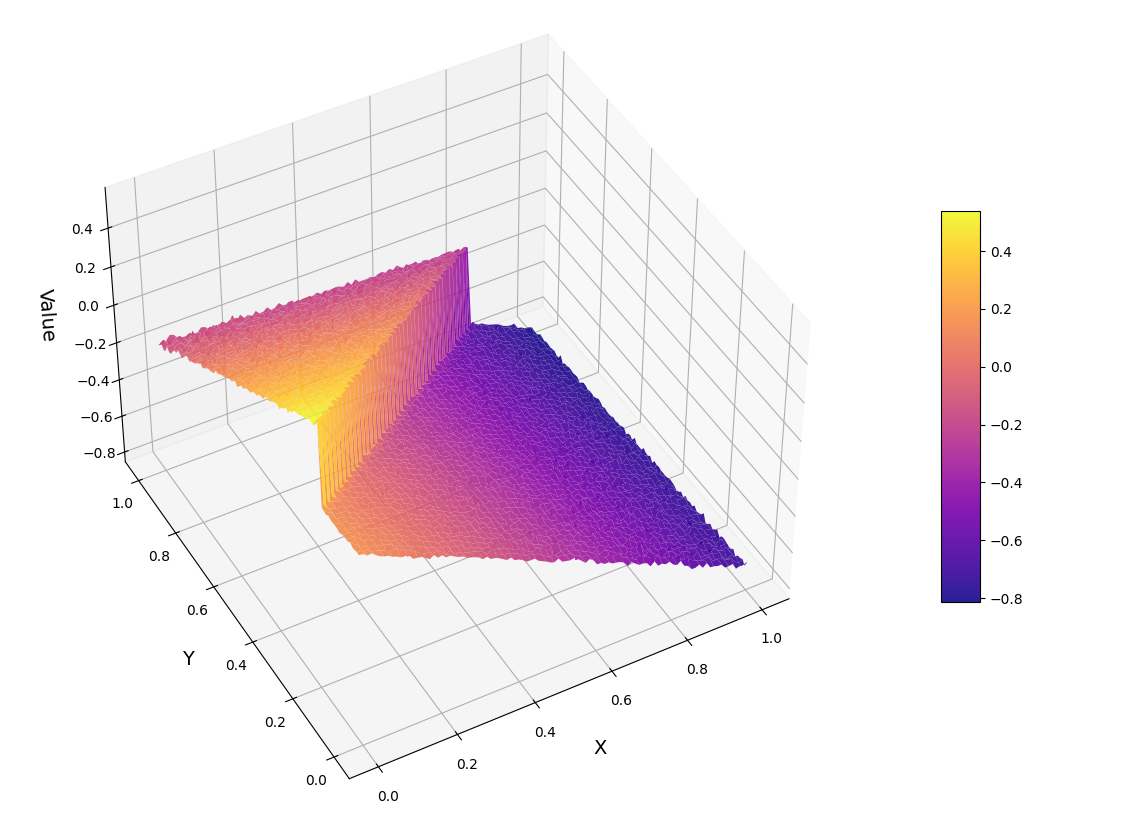} 
            \caption{}
            \label{jump_fig}
	\end{subfigure}
    \begin{subfigure}{0.3\linewidth}
            \centering
            \includegraphics[width=1.3\linewidth]{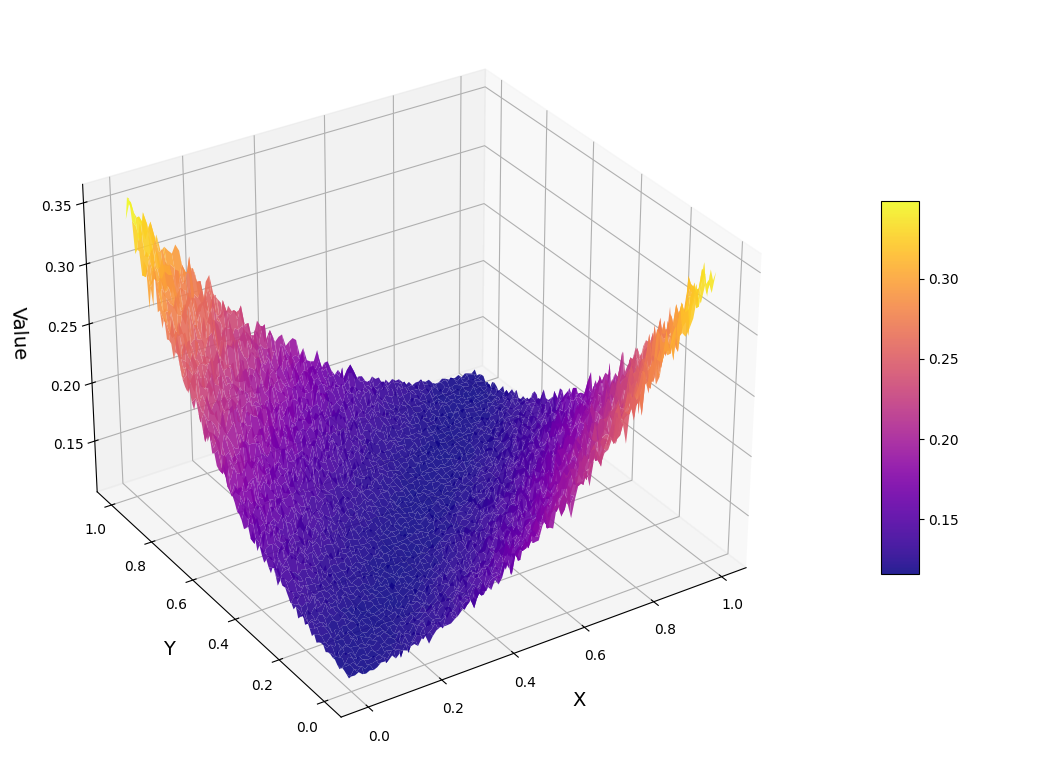}
            \caption{}
            \label{utility_var_fig}
    \end{subfigure}
    \begin{subfigure}{0.3\linewidth}
    		\centering
    		\includegraphics[width=1.3\linewidth]{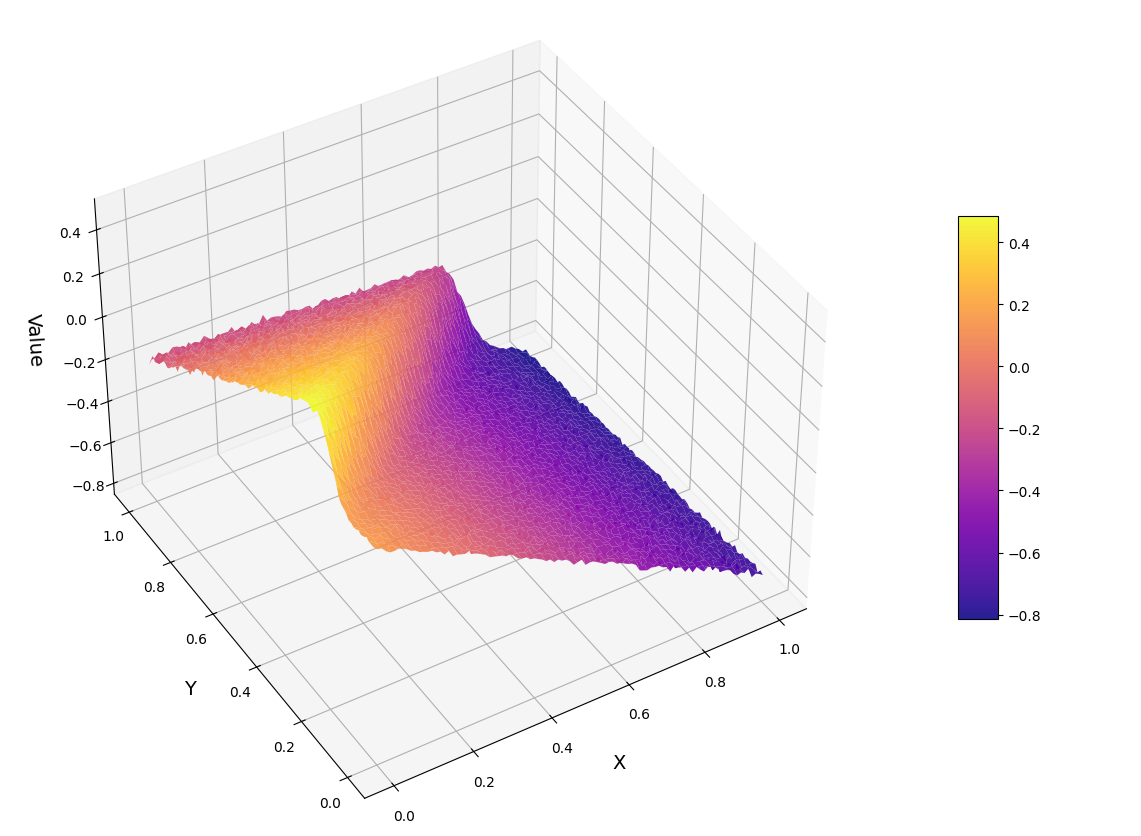} 
            \caption{}
            \label{no_jump_fig}
	\end{subfigure}
	\caption{\small \textit{(High-low example ~\ref{high_low_example})} (a) Visualization of ${U}^\mathrm{PA}$. (b) Visualization of the variance of \textit{heteroscedastic} noise $\sigma^2_{\varepsilon}$. (c) Geometry of $\Tilde{U}^\mathrm{PA}$ under the Assumption ~\ref{main_assumption}. }
\end{figure*}

The second technical challenge is the decision-variable-dependent randomness in the principal's reward structure. At round $t$, we can rewrite the principal's reward in \eqref{reward_eq} in the following form: 
\begin{align} \label{inp_rand}
y_t &= (v-x_t)_i, \quad i \sim p(a^\star(x_t)), \nonumber \\
    &= \sum_{j\in [m]} (v-x_t)_jp_j(a^\star(x_t)) + \varepsilon(x_t)  
     = U^{\mathrm{PA}}(x_t) + \varepsilon(x_t),
\end{align}
where $\varepsilon(x_t)$ is the zero-mean \textit{heteroscedastic} noise with variance $\mathbb{V}ar[\varepsilon(x_t)]= \sum_{j\in [m]} p_j(a^\star(x_t))[(v-x_t)_j- U^{\mathrm{PA}}(x_t)]^2$. Unlike homoscedastic models, the noise variance here depends on $x_t$ and is unknown to the decision maker. Figure~\ref{utility_var_fig} illustrates this input-dependent variance, $\sigma^2_{\varepsilon}$, for Example ~\ref{high_low_example}.
\section{Main Result I: Resolving Discontinuity of $U^{\mathrm{PA}}$}
Due to the aforementioned challenges, the authors in \cite{bacchiocchi2023learning,ho2014adaptive,zhu2023sample} prefer to discretize the contract space $\mathcal{X}$ and deal with a finite number of contracts. Instead, our approach is not to discretize $\mathcal{X}$, but to design algorithms that can work in a continuous contract space, thereby freeing us from the technicality of upper-bounding the discretization error. In this regard, we first make the following improvement to the agent's utility. 

Consider an embedding $\delta: \mathcal{A}\rightarrow \mathbb{R}^n$, and an $\mathbb{R}^n$-valued random variable $w$. We define the agent's best response by adding a stochastic counterpart to the agent's utility function, and update the optimal contract, respectively,
\begin{subequations}
    \begin{align} \label{modified_bp}
    \Tilde{a}^\star(x,w)& = \arg\max_{a \in \mathcal{A}} \left\{\tilde{U}^{\mathrm{A}}(x,a) = U^{\mathrm{A}}(x,a)+ \langle w,\delta(a)\rangle \right\},\\ \label{final_utility}
    \tilde{x}^\star &= \arg\max_{x \in \mathcal{X}} \left\{ \Tilde{U}^{\mathrm{PA}}(x)= \mathbb{E}_{w}[U^{\mathrm{P}}(x,\Tilde{a}^\star(x,w)) ] \right\},
   \end{align}
\end{subequations}
while the random feedback observed by the principal
\begin{align} \label{new_prod_fcn}
    \tilde{y} &= (v-x)_i , \quad i \sim q(x), \\ \nonumber
    \mathrm{where}\quad q_i(x)&= \mathbb{E}_w\left[p_i(\Tilde{a}^\star(x,w))\right] .   
\end{align}
Note, the updated production function $q(x)$ encodes the probability of any realization $i \in [m]$ by marginalizing over randomness related to $w$. Moreover, the regret in \eqref{reg_def} is updated as
\begin{equation} \label{reg_final}
    \Tilde{R}_T= T\Tilde{U}^{\mathrm{P}*} - \sum_{t=1}^T \Tilde{U}^{\mathrm{PA}}(x_t),
\end{equation}
where $\Tilde{U}^{\mathrm{P}*}= \max_{x \in \mathcal{X}} \:\: \Tilde{U}^{\mathrm{PA}}(x)$.

While the modification in \eqref{modified_bp} serves to restore the continuity of the principal's utility, it fundamentally reflects the inherent uncertainty of real-world agent behavior. In \cite{bacchiocchi2023learning,wang2023deep}, agent utilities are assumed to be static for a fixed $x$ and $a$. However, in practical settings like a crowdsourcing market, a worker's cost (effort) fluctuates due to unobserved latent factors such as fatigue, attention shifts, or external interruptions. Therefore, for the same contract and action pair $(x,a)$, the agent's utility could take different values for different instances of time. By introducing stochasticity into the agent's utility function, we effectively treat the principal's problem as an optimization over the \textit{robust expected behavior} of the agent, rather than overfitting to a deterministic model that inherently fails to capture such physical transients.

\begin{assumption}[Stochastic regularity] \label{main_assumption}
For the stochastic addition in the best response \eqref{modified_bp}, we assume the following: 
    \begin{enumerate}\renewcommand{\labelenumi}{(\roman{enumi})}
        \item \textit{One-to-one embedding:} For all $a_1,a_2 \in \mathcal{A}$, there exists a continuous embedding $\delta: \mathcal{A} \rightarrow \mathbb{R}^n$ such that $\delta(a_1)\neq \delta(a_2)$ if $a_1 \neq a_2$. 
        \item \textit{Full-dimensional support:} The random variable $w\in \mathbb{R}^n$ admits a probability density distribution that is absolutely continuous with respect to Lebesgue measure. In particular, for all $v_1\in \mathbb{R}^n, v_2 \in \mathbb{R}$, the probability of any hyperplane is $\mathbb{P}\left[w \in \mathbb{R}^n \mid \; \langle v_1,w \rangle=v_2 \right]=0$ (e.g., uniform distribution supported on $[0,1]^n$).
    \end{enumerate} 
\end{assumption}
\begin{example}[Embedding] 
If $\mathcal{A}$ is finite, then it suffices to choose the embedding image (i.e., $n=1$) as $\delta(a_i)= (i-1)/|\mathcal{A}|$. If $\mathcal{A}$ is a continuous set, a straightforward embedding is the identity map: $\delta(a)=a$.
To visualize an example of the proposed modification, consider $U^{\mathrm{A}}(x,a)= x- (a^4-2a^2)$ for $x=0$, and Figure~\ref{visual_main_assump} depicts how the optimizer $a^\star$ shifts to $\tilde{a}^\star$ as we add the identity embedding. 
\end{example}
\begin{figure}[htpb]
    \centering
    \begin{tikzpicture}
        
        \begin{axis}[
            name=plot1,
            width=7.5cm, height=5.0cm, 
            domain=-2:2, samples=100,
            axis lines=center,
            enlargelimits=true,
            xlabel={$a$},
            ylabel={$U^{\mathrm{A}}(x,a)$},
            xlabel style={anchor=north},
            ylabel style={anchor=east},
            ymin=-2, ymax=2.5,
            xtick=\empty,
            ytick={1}, 
            yticklabels={}, 
            title={}
        ]
            \addplot [blue, thick] {2*x^2 - x^4};
            
            \draw[dashed, red, thick] (axis cs:-2.2, 1) -- (axis cs:2.2, 1);
            
            \draw[dotted, thick, blue] (axis cs:1, 0) -- (axis cs:1, 1);
            \fill[blue] (axis cs:1, 1) circle (1.5pt);
            \node[below, blue, inner sep=3pt] at (axis cs:1, 0) {$a^\star$};
            
            \draw[dotted, thick, blue] (axis cs:-1, 0) -- (axis cs:-1, 1);
            \fill[blue] (axis cs:-1, 1) circle (1.5pt);
            
        \end{axis}

        \begin{axis}[
            name=plot2,
            at={(plot1.east)}, anchor=west, xshift=1.0cm, 
            width=7.5cm, height=5.0cm, 
            domain=-2:2, samples=100,
            axis lines=center,
            enlargelimits=true,
            xlabel={$a$},
            ylabel={$\tilde{U}^{\mathrm{A}}(x,a)$},
            xlabel style={anchor=north},
            ylabel style={anchor=east},
            ymin=-2, ymax=2.5,
            xtick=\empty,
            ytick={1.31}, 
            yticklabels={}, 
            legend pos=north east,
            title={}
        ]
            \addplot [purple, thick] {2*x^2 - x^4 + 0.3*x};
            \addlegendentry{$w = 0.3$}
            
            \draw[dashed, purple, thick] (axis cs:-2.2, 1.305) -- (axis cs:2.2, 1.305);
            
            \draw[dotted, thick, blue] (axis cs:1, 0) -- (axis cs:1, 1.3);
            \fill[blue] (axis cs:1, 1.3) circle (1.5pt);
            \draw[<-, blue, semithick] (axis cs:0.98, -0.05) -- (axis cs:0.5, -0.6) node[below] {$a^\star$};

            \draw[dotted, thick, purple] (axis cs:1.036, 0) -- (axis cs:1.036, 1.305);
            \fill[purple] (axis cs:1.036, 1.305) circle (1.5pt);
            \draw[<-, purple, semithick] (axis cs:1.056, -0.05) -- (axis cs:1.4, -0.6) node[below] {$\tilde{a}^\star$};

        \end{axis}

    \end{tikzpicture}
    \caption{Left: The agent's utility $U^{\mathrm{A}}(x,a)$ with optimizer marked at $a^\star$. Right: The modified utility $\tilde{U}^{\mathrm{A}}(x,a):= U^{\mathrm{A}}(x,a) + \langle w,a \rangle$ with a realization $w=0.3$. Angled pointers clearly distinguish the original optimizer $a^\star$ from the newly shifted optimizer $\tilde{a}^\star$.}
    \label{visual_main_assump}
\end{figure}
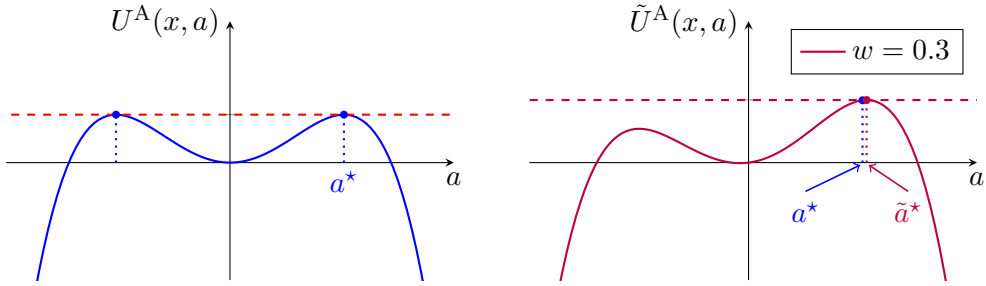

We also assume the following standard property for $U^{\mathrm{A}}$ and $U^{\mathrm{P}}$. 
\begin{assumption}[Compact and continuous domain]\label{cont_agent_assum}
    The agent's utility $U^{\mathrm{A}}: \mathcal{X}\times \mathcal{A}\rightarrow \mathbb{R}$ and the uncoupled principal's utility $U^{\mathrm{P}}: \mathcal{X}\times \mathcal{A}\rightarrow \mathbb{R}$ are continuous on their compact domains $\mathcal{X}$ and $\mathcal{A}$. 
\end{assumption}
The above assumptions allows us to characterize the following property of the correspondence $\tilde{a}^\star(x,w)$ in \eqref{modified_bp} and the principal's utility $\Tilde{U}^{\mathrm{PA}}$. 
\begin{theorem}[Continuous principal's utility]\label{main_lemma}
     Under Assumptions~\ref{main_assumption} and ~\ref{cont_agent_assum}, the following holds:
     \begin{enumerate}\renewcommand{\labelenumi}{(\roman{enumi})}
         \item The agent's best response $\tilde{a}^\star(x,w)$ in \eqref{modified_bp} is a singleton and continuous at any $x \in \mathcal{X}$ almost surely, i.e., for any $x \in \mathcal{X}$, we have
         \begin{align}
            \mathbb{P}\Big[ & w\in \Omega \mid   \; \tilde{a}^\star(x,w) \mathrm{\;is\;singleton\;and\;}  
             \forall a_n \in \tilde{a}^\star(x_n,w), \lim_{n \rightarrow \infty} a_n = \tilde{a}^\star(x,w) \Big]= 1.
        \end{align}
         \item The principal's utility $\Tilde{U}^{\mathrm{PA}}$ defined in \eqref{final_utility} is continuous on its domain.
     \end{enumerate}
\end{theorem}
\begin{proof}
The proof is decomposed into two parts:

\textit{  Part $(i)$.}
Define the value function
\[
V(x,w):= \max_{a \in \mathcal{A}} \left\{U^{\mathrm{A}}(x,a)+ \langle w,\delta(a)\rangle \right\}.
\]
Since $\mathcal{A}$ is compact and $U^{\mathrm{A}}$ is continuous on its domain, $V$ is finite everywhere. Moreover, for a fixed $x$, $V$ is convex in $w$ as it is the maximum of a family of affine functions. By an extension result from Danskin's theorem~\cite{bertsekas1971control}, we have
\[
\partial_w V(x,w)=\operatorname{co}\bigl\{ \delta(a)\mid a\in \tilde{a}^\star(x,w)\bigr\},
\]
where $\partial_w V(x,w)$ denotes the convex subdifferential of $V$ with respect to $w$, and
$\operatorname{co}\{\cdot\}$ denotes the convex hull of its input.
Consequently, we have
\[
V \text{ is differentiable at } x \text{ and } w
\;\;\Longleftrightarrow\;\;
\tilde{a}^\star(x,w) \text{ is a singleton}.
\]
From Rademacher's theorem \cite[Section $3.1$]{evans2025measure}, we know that if the function~$V$ is a finite convex function on $\mathbb{R}^n$, which is the case here, it is differentiable almost everywhere with respect to Lebesgue measure. Therefore, the set
\[
N:=\{w\in\mathbb{R}^n:\tilde{a}^\star(x,w)\text{ is not a singleton}\}
\]
has Lebesgue measure zero. Because the distribution of $w$ is absolutely continuous with respect to Lebesgue measure (following Assumption~\ref{main_assumption}), then
\begin{align} \label{unique_res}
 \mathbb{P}[w \in N]=0 
\quad\Longleftrightarrow\quad 
\mathbb{P}\bigl[\tilde{a}^\star(x,w)\text{ is a singleton}\bigr]=1   
\end{align}
By Assumptions~\ref{cont_agent_assum} and~\ref{main_assumption}, $(U^{\mathrm{A}}(x,a) + \langle w,\delta(a)\rangle)$ is jointly continuous in $x$ and $a$, which leads us to apply the celebrated Berge's maximum theorem \cite[Theorem $17.31$]{aliprantis2006infinite} on \eqref{modified_bp}. And we say that the correspondence $\Tilde{a}^\star(x,w)$ is upper-hemicontinuous at $x$. At the same time, we just proved in \eqref{unique_res} that $\Tilde{a}^\star(x,w)$ is unique. Thus, $\Tilde{a}^\star(x,w)$ is continuous on $x$. 

\textit{Part $(ii)$.} 
For any $x \in \mathcal{X}$, and any convergent sequence of $\{x_n\}\rightarrow x$, we have
\begin{align*}
    \lim_{x_n \rightarrow x} \Tilde{U}^{\mathrm{PA}}(x_n) = \lim_{x_n \rightarrow x} \mathbb{E}_{w}[U^{\mathrm{P}}(x_n,\Tilde{a}^\star(x_n,w)) ] 
    &= \mathbb{E}_w\left[\lim_{x_n \rightarrow x}U^{\mathrm{P}}(x_n,\Tilde{a}^\star(x_n,w))  \right] \\
    &= \mathbb{E}_w\left[U^{\mathrm{P}}(x,\Tilde{a}^\star(x,w))  \right]= \Tilde{U}^{\mathrm{PA}}(x) 
\end{align*}
where the second equality follows from the Dominated Convergence Theorem \cite[Theorem $1.19$]{evans2025measure}, as the function $U^{\mathrm{P}}$ is continuous (Assumption~\ref{cont_agent_assum}), and hence uniformly bounded over any compact set. The third equality holds true because of the continuity of $\tilde{a}^\star(x,w)$ in $x$ as shown in the proof of \textit{part} $(i)$. This concludes the proof. 
\end{proof}
It is noteworthy to mention that Theorem~\ref{main_lemma} is built using the abstract formulation in \eqref{agent_bp}, which implies that our result is also applicable to the class of repeated Stackelberg games \cite{lauffer2023no}. Any other utility structure besides the principal-agent model in \eqref{agent_utility} will follow Theorem ~\ref{main_lemma} as long as the necessary assumptions are met. To address the technical challenge of input-dependent randomness in \eqref{inp_rand}, we propose a dedicated algorithm in Section~\ref{hgp_section}.
\section{Main Result II: Tailored Kernel and Sub-linear Regret} \label{sec_kernel_algo}
In this section, we focus on designing an algorithm for the principal-agent model given by \eqref{agent_utility} and \eqref{reward_eq}. Since the previous work \cite{wang2023deep} has already established geometric characterizations of the principal's utility landscape with finite agent actions, we aim to exploit this to achieve a better regret rate. To this end, we restrict the complexity of our problem with the following assumption.   
\begin{assumption}[Finite action space] \label{first_assumption}
    The agent has a finite number of actions, i.e., $|\mathcal{A}|< \infty$.
\end{assumption}
By Assumption~\ref{first_assumption}, a candidate embedding can be $\delta(a_j)= (j-1)/|\mathcal{A}|$ for $a_j \in \mathcal{A}$, which leads to the agent's best response $\Tilde{a}^\star(x,w)= \arg\max_{a_j \in \mathcal{A}} \left\{\sum_{i\in [m]} (x)_i\;p_i(a_j)-c(a_j)+ \frac{w(j-1)}{|\mathcal{A}|} \right\},$ and the principal's utility, $\Tilde{U}^{\mathrm{PA}}(x)= \mathbb{E}_w \Big [ \sum_{i\in [m]}(v-x)_i \;p_i(\Tilde{a}^\star(x,w))\Big ]$. We now visualize $\Tilde{U}^{\mathrm{PA}}$ using Example~\ref{high_low_example} where $|\mathcal{A}|=2$. Assuming that $w \sim \mathcal{N}(0,\sigma_w^2)$, Figure~\ref{no_jump_fig} shows a continuous $\Tilde{U}^{\mathrm{PA}}$ as stated in Theorem~\ref{main_lemma}. In this regard, we identify two key observations: $(1)$ the surface of $\Tilde{U}^{\mathrm{PA}}$  consists of flat plateaus and steep transitions (see Figure~\ref{no_jump_fig} ), and $(2)$ the variance around $\Tilde{U}^{\mathrm{PA}}$ is still decision-variable-dependent as given in \eqref{inp_rand}. To address both observations, we consider stochastic bandits \cite{lattimore2020bandit} over a continuous domain. While parametric bandits \cite{rusmevichientong2010linearly} assume linear reward structures, akin to the simpler theory of linear contracts, we adopt a nonparametric approach using Gaussian Processes (GPs) to capture the complex structure of the principal's utility. 

A GP over the feasibility set $\mathcal{X}=[0,1]^m$, denoted by $GP_{\mathcal{X}}(\mu_t(\cdot),k(\cdot,\cdot))$, is a collection of random variables $(\Tilde{U}^{\mathrm{PA}}(x))_{x \in \mathcal{X}}$, such that every finite sub-collection of random variables $(\Tilde{U}^{\mathrm{PA}}(x_i))_{i=1}^n$ is jointly Gaussian with mean $\mu_t(x_i) = \mathbb{E}[\Tilde{U}^{\mathrm{PA}}(x_i)]$ and covariance $k(x_i,x_j) =  \mathbb{E}[(\Tilde{U}^{\mathrm{PA}}(x_i)-\mu_t(x_i))(\Tilde{U}^{\mathrm{PA}}(x_j)-\mu_t(x_j))]$, $1\leq i,j \leq n,n \in \mathbb{N}$. Algorithms use $GP_{\mathcal{X}}(0,k(\cdot,\cdot))$ as an initial prior distribution for the unknown function $\Tilde{U}^{\mathrm{PA}}$ over $\mathcal{X}$, where $k(\cdot,\cdot)$ is the kernel associated with a suitable Reproducing Kernel Hilbert Space (RKHS) $\mathcal{H}_k$. Algorithms such as GP-UCB \cite{srinivas2012information}, GP-PI \cite{hoffman2011portfolio}, and IGP-UCB \cite{chowdhury2017kernelized} leverage GPs by imposing regularity assumptions on the objective function $\Tilde{U}^{\mathrm{PA}}$: either (1) $\Tilde{U}^{\mathrm{PA}}$ is a sample from a GP prior, or (2) $\Tilde{U}^{\mathrm{PA}}$ resides in the RKHS $\mathcal{H}_k$ of a kernel $k$. Given the structured geometry of $\Tilde{U}^{\mathrm{PA}}$, the agnostic RKHS assumption is theoretically more robust than assuming the function is a random draw from a GP. Consequently, the focus is to identify a kernel $k$ whose RKHS $\mathcal{H}_k$ naturally models flat plateaus and steep transitions.
\subsection{Kernel design} \label{nn_kernel_sec}
Standard stationary kernels, such as the Gaussian kernel, are ill-suited for this problem as their reliance on a single, global lengthscale fails to adapt to the heterogeneous geometry of $\Tilde{U}^{\mathrm{PA}}$. While the underlying deterministic utility ($U^{\mathrm{PA}}$) exhibits distinct affine segments, our stochastic smoothing transforms them into a non-stationary landscape characterized by flat plateaus and steep transitions (Figure~\ref{no_jump_fig}). To mimic this structure, we employ the \textit{Arcsin kernel} (aka Neural Network kernel) \cite[Section $3$]{williams1998computation}, which arises as the limit of a single-hidden-layer Bayesian Neural Network with error function ($\text{erf}$) activation. Unlike stationary kernels, this kernel induces a non-stationary prior inherently designed to model sigmoidal superpositions. This inductive bias allows it to represent the soft thresholding behavior and varying local smoothness of the principal's expected~utility.

Consider the input-output mapping of a single-hidden-layer neural network $f_N:\mathbb{R}^n \rightarrow \mathbb{R}$ with an activation function $\phi: \mathbb{R}\rightarrow \mathbb{R} $, defined as  
\begin{subequations}
    \begin{align}
        \phi(z)& = \frac{2}{\sqrt{\pi}}\int_0^z \; \exp{(-t^2)} dt, \\
        f_N(x) &= \sum_{i=1}^N \varsigma_i \phi(\langle r_i, x\rangle + b_i),
    \end{align}
\end{subequations}
where the independent Gaussian priors $\varsigma_i \sim \mathcal{N}(0, \sigma_{\varsigma}^2/N)$, $r_i \sim \mathcal{N}(0, \sigma_v^2 \mathbf{I})$, and $b_i \sim \mathcal{N}(0, \sigma_b^2)$. The kernel function $k(x, y)$ is derived as the limiting covariance
\begin{align*} 
    k(x, y) &:= \lim_{N \to \infty} \mathbb{E}[f_N(x) f_N(y)]= \sigma_{\varsigma}^2 \mathbb{E}_{z_x, z_y}[\phi(z_x) \phi(z_y)],
\end{align*}
where the last equality follows by applying the central limit theorem,  $z_x = \langle r,  x\rangle + b$ and $z_y = \langle r,y\rangle + b$ are jointly Gaussian variables. The expectation in the above equation admits a closed-form solution known as the Neural Network (aka `arcsin') kernel \cite{williams1998computation}, described by
\begin{align} \label{kernel_eq}
    & k(x, y)  
     =\frac{2\sigma_{\varsigma}^2}{\pi} \arcsin \left( \frac{\sigma_v^2 \langle x, y \rangle + \sigma_b^2}{\sqrt{(1 + \sigma_v^2 \langle x, x \rangle + \sigma_b^2) (1 + \sigma_v^2 \langle y, y \rangle + \sigma_b^2)}} \right),
\end{align}
where $\sigma^2_v$ and $\sigma^2_b$ are treated as hyper-parameters. Apart from being symmetric, the positive definiteness of this kernel is guaranteed by its construction as an expectation of squares, $\sum_{i,j} c_i k(x_i,x_j)c_j = \sigma^2_w \mathbb{E}\left[\left(\sum_i c_i\phi(z_{x_i})\right)^2\right] \geq 0$. By the Moore-Aronszajn theorem, any symmetric and positive definite kernel on $\mathcal{X}$ induces a unique RKHS $\mathcal{H}_k$ defined as
\begin{align*}
    \mathcal{H}_k& := \Big\{ h \in \mathbb{R}^{\mathcal{X}} \mid \exists \alpha_i\in \mathbb{R}, x_i \in \mathcal{X}\; \forall i \in \mathbb{N} \; \mathrm{with}\; 
     h(x)= \sum_{i=1}^\infty \alpha_i k(x_i,x)\; \mathrm{and}\; \sum_{i=1}^\infty \sum_{j=1}^\infty \alpha_i k(x_i,x_j)\alpha_j < \infty   \Big\}.
\end{align*}
To empirically validate the correct hypothesis space of $\Tilde{U}^{\mathrm{PA}}$, we run a kernel regression on an instance of `High-low' example (similar to Example~\ref{high_low_example}) with $k(x,y)$ in \eqref{kernel_eq}. We plot the learned function in Figure~\ref{nn_kernel_func}, which clearly shows the flat plateaus with sigmoidal transitions we were looking for. Therefore, the $\mathcal{H}_{k}$ induced by $k(x,y)$ in \eqref{kernel_eq} is a candidate hypothesis space for $\Tilde{U}^{\mathrm{PA}}$. 
\begin{figure}
    \centering
    \includegraphics[width=0.6\linewidth]{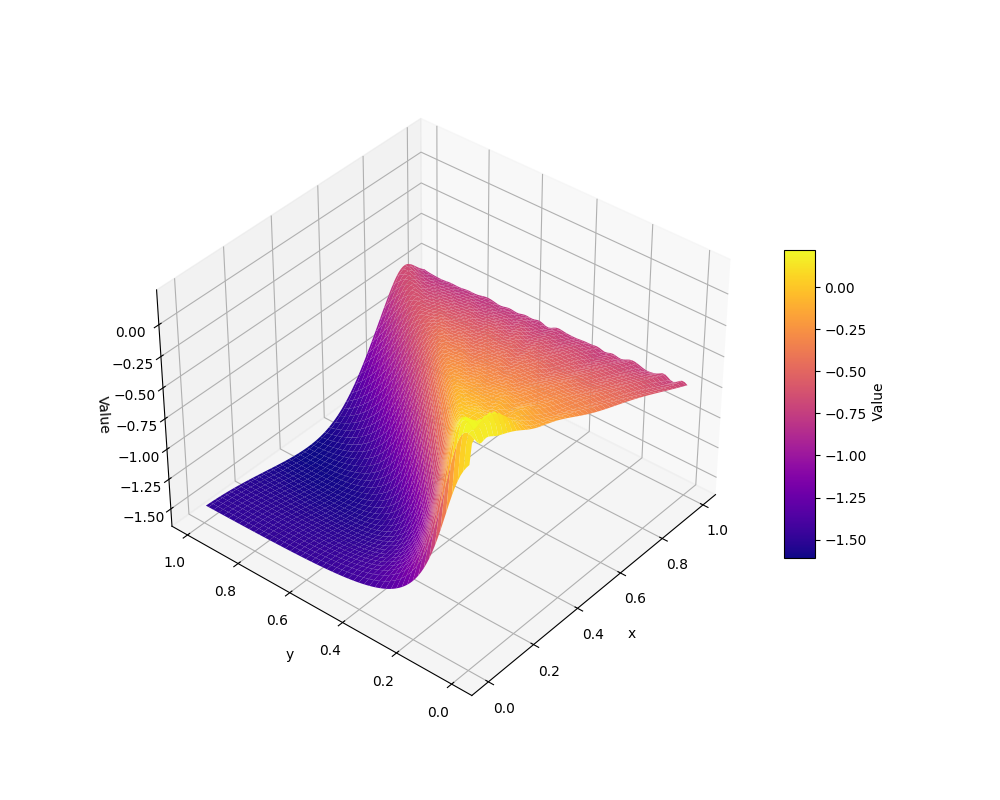}
    \caption{Representative function of $\Tilde{U}^{\mathrm{PA}}$, which is built with Neural Network kernel using kernel regression where the empirical samples are generated from an instance of `High-low' example~\ref{high_low_example}.}
    \label{nn_kernel_func}
\end{figure}
\subsection{Description of the algorithm} \label{hgp_section}
First, we develop methods to address the \textit{heteroscedastic} noise in the principal's utility, which is denoted as a technical challenge in \eqref{inp_rand}. In particular, we learn the unknown variance of this noise using feedback from repeated sampling of the principal's utility at the same input. A good estimate of the variance helps us obtain a more accurate upper confidence bound for the unknown $\Tilde{U}^{\mathrm{PA}}$. Then we detail our main algorithm, based on a upper confidence bound-based acquisition function constructed using the Neural Network kernel in \eqref{kernel_eq}. To start with, similar to \eqref{inp_rand}, we represent the feedback as a noisy version of $\tilde{U}^{\mathrm{PA}}$, and write
\begin{equation} \label{feedback_eq}
    \tilde{y}_t= \tilde{U}^{\mathrm{PA}}(x_t) + \tilde{\varepsilon}(x_t),
\end{equation}
where $\tilde{\varepsilon}(x_t)$ is assumed to be a zero-mean \textit{heteroscedastic} noise with the variance $\mathbb{V}ar(\tilde{\varepsilon}(x_t))= \sum_{i\in [m]} q_i(x_t)[(v-x_t)_i- \Tilde{U}^{\mathrm{PA}}(x_t)]^2 $, and the production function $q(x_t)$ is defined in \eqref{new_prod_fcn}. It is straightforward to follow that $\mathbb{V}ar(\tilde{\varepsilon}(x_t))$ is continuous in $x_t$ since the continuity of $\tilde{a}^\star(x_t,w)$ in $x_t$ (Theorem~\ref{main_lemma}) makes $q_i(x_t)$ to be continuous on its argument. Here, we first impose the following assumption on $\tilde{\varepsilon}(x_t)$ to restrict it within a well-defined function class. 
\begin{assumption} [Noise function class] \label{var_proxy_assumption}
    The noise variance $\mathbb{V}ar( \tilde{\varepsilon}(x_t))$ belongs to an RKHS induced by some kernel $\kappa^v$, i.e., $\mathbb{V}ar( \tilde{\varepsilon}(\cdot))\in \mathcal{H}_{\kappa^v}$ with its RKHS norm is upper bounded with $B_{v}>0$.
\end{assumption}
Such assumptions are quite common when dealing with \textit{heteroscedastic} noise \cite{kirschner2018information,makarova2021risk}. To learn $\mathbb{V}ar(\tilde{\varepsilon}(\cdot))$, we construct a repeated setting, where we collect $\ell>1$ samples for each $x_t$ and form the feedback collection $\{\tilde{y}_{t}^i\}_{i=1}^\ell$. Then, we evaluate the sample mean and variance of $\tilde{\varepsilon}(x_t)$ as
\begin{subequations} \label{empirical_y}
    \begin{align}
        \hat{y}_t &= \frac{1}{\ell}\sum_{i=1}^\ell \tilde{y}_{t}^i= \frac{1}{\ell}\sum_{i=1}^\ell(\tilde{U}^{\mathrm{PA}}(x_t)+ \tilde{\varepsilon}_i(x_t)), \\
        \hat{v}_t &= \frac{1}{\ell-1} \sum_{i=1}^\ell (\tilde{y}_{t}^i-\hat{y}_t)^2.
    \end{align}
\end{subequations}
While we have random feedback samples of $\tilde{y}_t$, we represent the randomness of sample variance using a zero-mean noise $\eta(x_t)$. Since $(v)_i,(x)_i \in [0,1]$, it is straightforward to check that the noise $\tilde{\varepsilon}(x_t)$ is uniformly bounded as $|\tilde{\varepsilon}(x_t)|<\sqrt{\chi_u}$, which leads to its variance function $\mathbb{V}ar( \tilde{\varepsilon}(x_t)) < \chi_u$ with $\chi_u>0$. And this allows us to further state $\hat{v}_t=  \mathbb{V}ar( \tilde{\varepsilon}(x_t)) + \eta(x_t)$. Moreover, the noise $\eta(x_t)$ adheres to the following standard assumption. 
\begin{assumption} [Noise process regularity]
\label{noise_assumption}
    The noise $\eta(x_t)$ is $\rho_{\eta}(x_t)$-sub-Gaussian with known $\rho^2_{\eta}(x_t)$. Moreover, the process~$\{\eta(x_t)\}_{t\geq 1}$ is generated independently in time.
\end{assumption}
The sub-Gaussianity assumption of $\eta(x_t)$ naturally follows from uniform boundedness of $\tilde{\varepsilon}(x_t)$. Following a brief introduction of GP models at the start of Section~\ref{sec_kernel_algo}, the posterior GP mean ($\mu_t(\cdot)$) and variance ($\sigma^2_t(\cdot)$) of $\Tilde{U}^{\mathrm{PA}}$ are respectively calculated based on the previous measurements $\hat{y}_{1:t}= [\hat{y}_1, \dots \hat{y}_t]^{\top}$
        \begin{equation} 
        \begin{aligned}
        \mu_t(x) &= k_t(x)^{\top} (K_t + \lambda \Sigma_t)^{-1} \hat{y}_{1:t}, \\
        \sigma^2_t(x) &= \frac{1}{\lambda}(k(x,x)-k_t(x)^{\top}(K_t +\lambda\Sigma_t)^{-1}k_t(x) ),
            \end{aligned} \label{post_mean_var}
\end{equation}
where $\Sigma_t:= \text{diag}( \mathbb{V}ar( \tilde{\varepsilon}(x_1)), \dots,  \mathbb{V}ar( \tilde{\varepsilon}(x_t)))$, $(K_t)_{i,j}= k(x_i,x_j)$ is an element of the covariance matrix $K_t$, $\lambda>0$ is a parameter trading off the magnitude of prior variance, and $k_t(x):= [k(x_1,x), \dots , k(x_t,x)]^{\top}$. Here, $k(x_i,x_j)$ refers to our proposed Neural Network kernel in \eqref{kernel_eq}. As $\Sigma_t$ depends on the unknown variance $ \mathbb{V}ar( \tilde{\varepsilon}(x_t))$, we construct a separate GP model to learn it. Since we could not decode any specialized geometry for $\tilde{\varepsilon}(\cdot)$, we consider a universal kernel, such as a Gaussian kernel, to define the GP. The corresponding posterior mean $\mu^v_t$ and variance $\sigma^v_t$ are calculated based on \eqref{post_mean_var} by replacing $\Sigma_t$ with $\Sigma^v_t:= \text{diag}[\rho_{\eta}^2(x_1), \dots, \rho_{\eta}^2(x_t)]$, $\hat{y}_{1:t}$ with $\hat{v}_{1:t}= [\hat{v}_1, \dots , \hat{v}_t]^{\top}$, and using Gaussian kernel $k^v(\cdot,\cdot)$. Then, we build the upper confidence bound $\text{ucb}_t(\cdot)$ as 
\begin{equation} \label{ucb_v}
    \text{ucb}_t (x):= \mu^v_{t-1} (x) + \beta^v_t \sigma^v_{t-1}(x),
\end{equation}
where the exploration-exploitation trade-off parameter $\beta^v_t= \sqrt{2\log \left(\frac{\det(\lambda \Sigma^v_t + K^v_t)^{0.5}}{\delta \det (\lambda \Sigma^v_t)^{0.5}}\right)} + \sqrt{\lambda}B_v $, taken from \cite[Lemma $7$]{kirschner2018information}. We can approximate $\Sigma_t$ with the upper confidence bound of $ \mathbb{V}ar( \tilde{\varepsilon}(x_t))$ as 
\begin{equation} \label{sigma_hat}
    \hat{\Sigma}_t = \text{diag}(\min \{\text{ucb}_t(x_1), \chi_u\}, \dots , \min \{\text{ucb}_t(x_t), \chi_u\})/\ell.
\end{equation}
Now we are in a position to define the upper confidence bound-based acquisition function for $\tilde{U}^{\mathrm{PA}}$ using the current posterior mean $\mu_{t-1}$ and variance $\sigma^2_{t-1}$ as 
\begin{equation} \label{aq_func}
    x_t = \arg\max_{x \in \mathcal{X}} \:\: \underbrace{\mu_{t-1}(x) + \beta_t \sigma_{t-1}(x)}_{\mathrm{Acquisition\; function}},
\end{equation}
where $x_t$ is the action of the algorithm to maximize the acquisition function, and the exploitation-exploration trade-off parameter
\begin{align} \label{beta_t}
    \beta_t= \sqrt{2\log \left(\frac{\det(\lambda \hat{\Sigma}_t + K_t)^{0.5}}{\delta \det (\lambda \hat{\Sigma}_t)^{0.5}}\right)}+ \sqrt{\lambda}B,
\end{align}
which follows from \cite[Lemma $7$]{kirschner2018information}. In \eqref{beta_t}, $\delta$ is the confidence interval, and $B$ is the upper-bound of the RKHS-norm of its kernel. Intuitively, $\beta_t$ determines how much we trust the accuracy of the current posterior mean ($\mu_{t-1}$), which is an approximation of $\Tilde{U}^{\mathrm{PA}}$ on the fly, by suitably adding the posterior variance ($\sigma_{t-1}$), since the posterior variance quantifies the uncertainty of $\mu_{t-1}$ being the true unknown $\Tilde{U}^{\mathrm{PA}}$. Note, the role of $\hat{\Sigma}_t$ in \eqref{beta_t} is to better estimate the unknown variance matrix (i.e., $\Sigma_t$), which essentially improves the exploitation-exploration trade-off. Needless to say, the posterior mean and variance of the acquisition function in \eqref{aq_func} is built with our proposed Neural Network kernel in \eqref{kernel_eq}. In Algorithm~\ref{algo}, we explain the operation of our proposed methodology. Note that the initial dataset $\mathcal{X}_0$ is created by sampling any fixed $x_0\in\mathcal{X}$ for $\ell$ times.  
\begin{algorithm}
\caption{\texttt{Heteroscedastic GP-UCB}}
\label{algo}
\SetKwInOut{Input}{Input}
\SetKwInOut{Output}{Output}

\textbf{Initialization:} $B,B_v$, $\lambda$, $\ell$, Kernel functions $k(\cdot,\cdot)$, $k^v(\cdot,\cdot)$, Confidence level $\delta$, Initial dataset $\mathcal{X}_0=\{(x^i_0,\tilde{y}^i_0)_{i \in [l]}\}$, and Prior $\mu_0=\mu^v_0=0$. 

\For{$t= 1,2, \dots $}{
    Compute $\hat{y}_t$ and $\hat{v}_t$ following \eqref{empirical_y} (use $\mathcal{X}_0$ if $t=1$). \\
    Update $\text{ucb}_t(\cdot)$ following \eqref{ucb_v}. \\
    Update $\hat{\Sigma}_t$ following \eqref{sigma_hat}. \\
    Update the posterior $\mu_t(\cdot)$ and $\sigma^2_t(\cdot)$ following \eqref{post_mean_var}. \\
    Select $x_t= \arg\max_{x \in \mathcal{X}} \:\: \mu_{t-1}(x) + \beta_t \sigma_{t-1}(x)$
    
    \For{$i \in \{1, \dots, \ell\}$}{   
           Collect $\{\tilde{y}_{t}^i\}$, where $\tilde{y}_{t}^i= \tilde{U}^{\mathrm{PA}}(x_t) + \tilde{\varepsilon}_i(x_t)$
    }
     
}
\end{algorithm}
\subsection{Regret bound}
The asymptotic performance of \texttt{Heteroscedastic GP-UCB} is based on the regret bound of \cite[Theorem $1$]{makarova2021risk} where the authors provide 
\begin{align}
    \mathbb{P}\left[ \Tilde{R}_T \leq \beta_T \ell \sqrt{\frac{2 \hat{\gamma}_T T}{\log(1+\frac{\ell}{\chi_u^2})}} \right] \geq 1-\delta,
\end{align}
where $\delta$ is the confidence level, $\beta_T$ is defined in \eqref{beta_t}, and $\hat{\gamma}_T$ denotes the maximum information gain \cite{srinivas2012information}. At time $T$, $\hat{\gamma}_T$ is defined using mutual information $I(y_{1:T},\Tilde{U}^{\mathrm{PA}}_{1:T})$ \cite[Equation $2.28$]{cover1999elements} between the feedback $y_{1:T}=[y_1, \dots y_T]^\top$ and $\Tilde{U}^{\mathrm{PA}}_{1:T}=[\Tilde{U}^{\mathrm{PA}}(x_1), \dots \Tilde{U}^{\mathrm{PA}}(x_T)]^\top$ as
\begin{align}
    \hat{\gamma}_T:= \; \max_{A \subset \mathcal{X},|A|=T} \quad I(y_{1:T},\Tilde{U}^{\mathrm{PA}}_{1:T}).
\end{align}
In other words, maximum information gain is essentially the maximum uncertainty reduction of the unknown function $\Tilde{U}^{\mathrm{PA}}$ for a set of $T$ sampled data (say, $A$) in $\mathcal{X}$. The main technical challenge in this setup is to find an upper bound for $\hat{\gamma}_T$. $\hat{\gamma}_T$ is kernel-specific, and there exist standard upper bounds for $\hat{\gamma}_T$ where the kernel $k$ is Gaussian, Matern, or Linear \cite{srinivas2012information}. However, to our best knowledge, no such upper bound exists for our proposed neural-network kernel in \eqref{kernel_eq}. Following \cite[Theorem $8$]{srinivas2012information}, the upper bound of $\hat{\gamma}_T$ is a function of the eigenvalue tail sum $B_k(T_*)$ of the respective kernel, where $B_k(T_*):= \sum_{s=T_*+1}^{\infty}\lambda_s$, and $\lambda_s$ is the eigenvalue corresponding to the integral operator of any kernel $k(x,y)$. For our proposed neural network kernel in \eqref{kernel_eq}, the following result provides an upper bound for $B_k(T_*)$.  
\begin{proposition}[Bound on eigenvalue tail sum] \label{eigen_val_lemma}
    Given the neural network kernel $k(x,y)$ in \eqref{kernel_eq}, its eigenvalues $\lambda_s$, corresponding to the integral operator, decay exponentially. Consequently, the eigenvalue tail sum $B_k(T_*)= \sum_{s=T_*+1}^{\infty} \lambda_s$ is bounded as
    \begin{align}
        B_k(T_*) \leq \mathcal{O}\big(\exp{{(-\beta T^{1/m}_*)}}\big)
    \end{align} 
    for some constant $\beta >0$.
\end{proposition}
The proof is given in Appendix~\ref{eigen_proof}. The proof of the Proposition~\ref{eigen_val_lemma} hinges on the analytic property of any kernel. The analytic property refers to a certain level of smoothness of any kernel. It turns out that our proposed kernel enjoys this property (Lemma~\ref{analytic_lemma}) which leads to its exponentially decaying eigenvalues. With the help of Proposition~\ref{eigen_val_lemma}, we manage to derive an explicit upper bound for $\hat{\gamma}_T$, which helps us establish the main regret bound of this study.
\begin{theorem}[Sub-linear regret] \label{main_theorem}
    Consider $\tilde{U}^{\mathrm{PA}} \in \mathcal{H}_k$ with bounded RKHS-norm (i.e., $\|\tilde{U}^{\mathrm{PA}}\|_{\mathcal{H}_k} \le B$) and sampling model in \eqref{feedback_eq} with unknown variance $\mathbb{V}ar( \tilde{\varepsilon}(x_t))$ that satisfies Assumptions~\ref{var_proxy_assumption} and~\ref{noise_assumption}. Let $\{x_t\}_{t=1}^{T}$ denote the set of actions chosen by Algorithm~\ref{algo}. Setting $\lambda=1$ in \eqref{beta_t}, the cumulative regret $\Tilde{R}_T$ defined in \eqref{reg_final} is bounded as 
    \begin{equation}
        \mathbb{P}\left[\Tilde{R}_T \leq \mathcal{O}\left(\sqrt{T}(\log T)^{m+1}\right)\right] \geq 1 - \delta
    \end{equation}
    for a chosen confidence level $\delta >0$. 
\end{theorem}
\begin{proof}
The broad idea of the proof is to derive a kernel-specific upper bound of the information gain for an already existing regret bound in \cite{makarova2021risk}. To bound the information gain, we follow key ideas from \cite{srinivas2012information} to bound the eigenvalue tail sum of our proposed kernel in \eqref{kernel_eq}, and Proposition~\ref{eigen_val_lemma} does this job. To prove Proposition~\ref{eigen_val_lemma}, we require the analytic property of our kernel, which is given by Lemma~\ref{analytic_lemma}. The rest of the proof uses a few known algebraic manipulations to simplify the regret bound. 

We first write from Theorem $1$ of \cite{makarova2021risk} for  the coefficient of risk $\alpha =0$
\begin{align}
    \mathbb{P}\left[ \Tilde{R}_T \leq \beta_T \ell \sqrt{\frac{2 \hat{\gamma}_T T}{\log(1+\frac{\ell}{\chi_u^2})}} \right] \geq 1-\delta,
\end{align}
where $\delta$ is the confidence level, and the maximum information gain \cite[Equation $42$]{makarova2021risk} is defined as
\begin{equation} \label{inf_gain}
    \hat{\gamma}_T = \max_{A \subset D, |A|=T} \:\: \sum_{t=1}^{T} 0.5 \log \left( 1 + \frac{\sigma^2_{t-1}(x_t | \text{diag}(\chi_u^2/\ell))}{\chi^2_u /\ell} \right).
\end{equation}
One could compare \eqref{inf_gain} with the information gain of homoscedastic noise ($\gamma_T$ in \cite[Lemma $5.3$]{srinivas2012information}) and state that $\hat{\gamma}_T$ is simply same as $\gamma_T$ with a constant noise variance $\sigma^2= (\chi_u^2/\ell)$. Therefore, making the  required modification and replacing the upper bound of $\beta_T$ from equation $(25)$ in \cite{makarova2021risk}
\begin{equation} \label{reg_1}
    \Tilde{R}_T \leq \frac{\ell\sqrt{T}}{\sqrt{\log(1+ \frac{\ell}{\chi_u^2})}} \left( \sqrt{4 \gamma_T\log(1/\delta) + 2 \gamma_T^2} + B\sqrt{2\lambda \gamma_T} \right).
\end{equation}
Now, the key part is to upper bound $\gamma_T$ for our proposed kernel $k(x,y)$ in \eqref{kernel_eq}. For this, we use the result in \cite[Theorem $8$]{srinivas2012information} as
\begin{align} \label{lambda_T}
    &\gamma_T \leq \frac{0.5}{1-e^{-1}} \max_{r \le T} \:\: \Big(T_* \log(rn_T/\sigma^2)+ 
     C_4\sigma^{-2}(1-r/T)(\log T)(T^{\tau +1}B_k(T_*)+1) \Big) + \mathcal{O}(T^{1 - {\tau \over m}})
\end{align}
for any $T_* \in \{1, \dots , n_T\}$ and $\tau>0$. Here $B_k(T_*)= \sum_{s>T_*} \lambda_s$ with $\{\lambda_s\}$ being the operator spectrum of the kernel $k(x,y)$ with respect to uniform distribution over $D$, $n_T= C_4T^{\tau}\log T$ with $C_4= 2\mathcal{V}(D)(2\tau+1)$, and $\mathcal{V}(D)= \int_{x \in D}dx$. 

Absorbing \eqref{lambda_T} inside `Big-O' notation, we write 
\begin{align*}
    \gamma_T &\leq \mathcal{O}\Big(\max_{r=1,\dots, T}\: (T_*\log(rn_T)) + (T-r)T^{\tau}B_k(T_*)\Big) + \mathcal{O}(T^{1- \tau/m}). 
\end{align*}
Now setting $\tau=m$ converts $\mathcal{O}(T^{1-m/m})= \mathcal{O}(1)$. As $T \rightarrow \infty$, we can further compress the upper bound and write 
\begin{align}
    \gamma_T  & \leq \mathcal{O}\Big(\max_{r=1,\dots, T}\: (T_*\log(rn_T)) + T^{m+1}B_k(T_*)\Big), \nonumber \\
    &= \mathcal{O}\Big(\max_{r=1,\dots, T}\: (T_*\log(rT^m)) + T^{m+1}B_k(T_*)\Big), \quad (\text{by setting $n_T=\mathcal{O}(T^m\log T)$}), \nonumber\\
    & = \mathcal{O}\Big(T_*\log (T^{m+1}) + T^{m+1}B_k(T_*)\Big), \quad \text{(max is attained at $r=T$)}, \nonumber \\
    &= \mathcal{O}\Big(T_*\log (T) + T^{m+1}B_k(T_*)\Big).
\end{align}
From Proposition~\ref{eigen_val_lemma}, we know that $B_k(T_*) \leq \mathcal{O}(e^{-\beta T^{1/m}_*})$. Replacing it in the above equation
\begin{align} \label{gama_ub}
    \gamma_T & \leq \mathcal{O}\Big(T_*\log T + T^{m+1}e^{-\beta T^{1/m}_*}\Big), \nonumber \\
    & = \mathcal{O}\Big((\log T)^m \log T\Big), \:\:(\text{by setting $T_*= \Big(\frac{m+1}{\beta}\log T\Big)^m$, }\text{the second term is $\mathcal{O}(1)$}) \nonumber\\
    &= \mathcal{O}\Big((\log T)^{m+1}\Big).
\end{align}
Now, simplifying \eqref{reg_1} in terms of $\gamma_T$ 
\begin{align}
    \Tilde{R}_T & \leq \mathcal{O}\left(\sqrt{T}(\sqrt{\gamma_T + \gamma_T^2} + \sqrt{\gamma_T})\right), \nonumber \\
    &= \mathcal{O}\left(\sqrt{T}(\gamma_T +\sqrt{\gamma_T}) \right) \nonumber \text{(as $\gamma_T^2$ is dominant over $\gamma_T$)},\\
    &= \mathcal{O}\left(\sqrt{T}(\gamma_T ) \right) \nonumber \text{(as $\gamma_T$ is dominant over $\sqrt{\gamma_T}$)},\\
    &= \mathcal{O}\Big(\sqrt{T}(\log T)^{m+1}\Big) \:\: (\text{replacing equation \eqref{gama_ub}}).
\end{align}
Therefore, we have $\mathbb{P}\left[\Tilde{R}_T \leq \mathcal{O}\left(\sqrt{T}(\log T)^{m+1}\right)\right] \geq 1 - \delta,$ where $\delta> 0$ is the user-defined confidence level. This concludes the proof. 
\end{proof}
\section{Application: V2G Incentive Design }
In a standard V2G ecosystem, an Aggregator coordinates a fleet of electric vehicle (EV) owners to trade energy with the Distribution System Operator (DSO). In particular, Aggregators provide ancillary services to local DSOs in response to time-varying tariffs ($\lambda^g$), while procuring them from EV users. One key conflict arises because discharging power to the grid accelerates battery degradation of the EVs. This degradation imposes a private, stochastic cost on the EV owner, driven by the unobservable decline of battery health. In addition, the EV users impose a subjective cost for ``range anxiety'' while participating in V2G. Because the Aggregator cannot directly observe these individual users' costs or forcefully command the EVs to discharge, they must rely on incentives. 
\subsection{Mapping V2G to principal-agent model } 
To model the interaction between the Aggregator and EV users, we first clarify the research goal to be addressed in this application. To this end, we ask the following research question. 
\begin{figure}
    \centering
    \includegraphics[width=0.7\linewidth]{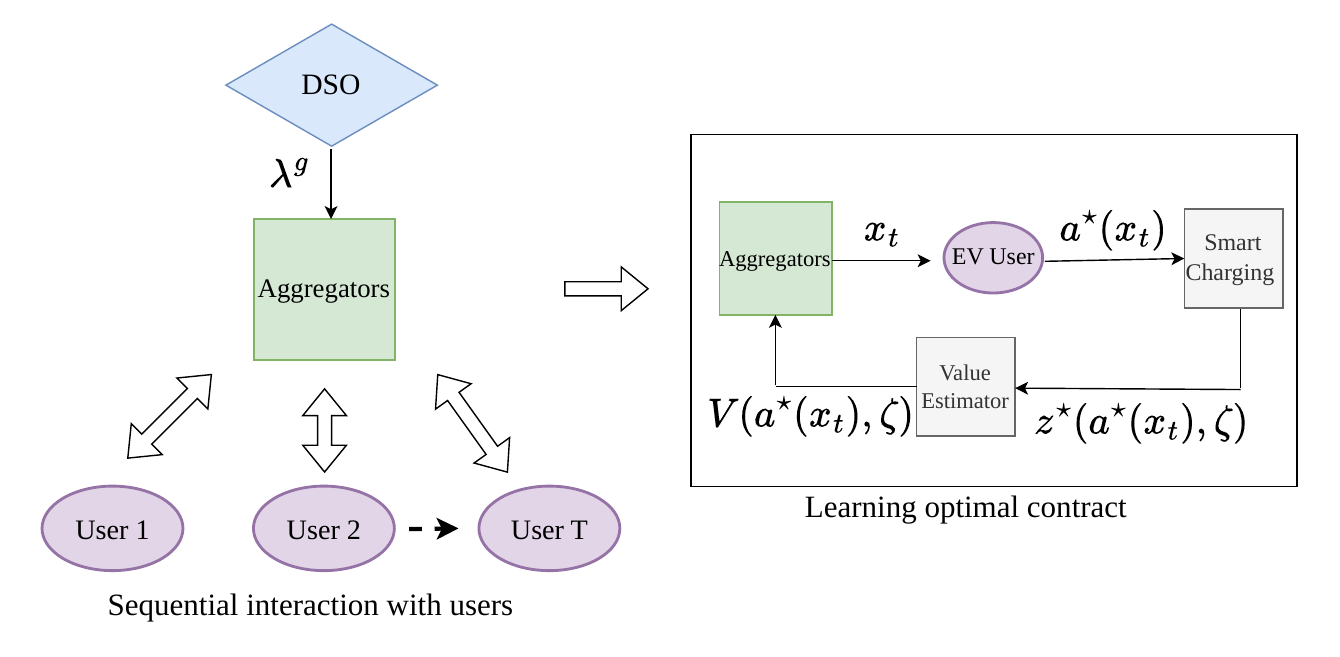}
    \caption{V2G incentive design schematic showing sequential interaction between aggregators and EV users.}
    \label{v2g_schematic}
\end{figure}
\begin{problem}
    How can an aggregator dynamically learn an adaptive incentive scheme (contract) that elicits optimal EV participation to maximize the Aggregator's revenue, despite the unobservable and stochastic nature of user costs?
\end{problem}
The solution approach is to map the V2G ecosystem into our dynamic principal-agent framework depicted in Figure~\ref{pa_fig}. In particular, we want to mathematically model the interaction between the Aggregator and EV user in the form of \eqref{agent_utility}, where the Aggregator is the principal and the EV user is the agent. As depicted in Figure~\ref{v2g_schematic}, we model the incentive design problem as a sequential interaction between the Aggregator and an EV user, where at time $t$, the aggregator offers a contract $x_t:=[x^\mathrm{l},x^\mathrm{h}] \in [0,1]^2$, specifying monetary incentives for low ($\mathrm{l}$) and high ($\mathrm{h}$) V2G contribution outcome set $\{\mathrm{l}, \mathrm{h}\}$. The EV user responds with a decision $a(x_t) \in \mathcal{A}:= \{0,1\}$, where $0$ denotes participation and $1$ denotes non-participation in the V2G program. Assume there exists a utility function $U^{\mathrm{A}}(x_t,a_t):=\sum_{i\in [2]} (x_t)_i\;p_i(a_t)-c(a_t) $ for the EV user, where $c:\mathcal{A}\rightarrow\mathbb{R}$ denotes the user's cost function, and $p:\mathcal\{0,1\}\rightarrow \Delta_2$ is the production function. The EV user optimizes $U^{\mathrm{A}}$ to retrieve the best response $a^\star(x_t)$. Based on $a^\star(x_t)$, an outcome $\{l,h\}\ni i_t \sim p(a^\star(x_t))$ is realized. However, the process from the EV user's optimal action $(a^\star(x_t))$ to the realization of an outcome $(i_t)$ includes several steps within the context of V2G operation. In particular, the non-trivial part is to characterize a production function $p(\cdot)$ that links EV users' best response $a^\star(x_t)$ to a stochastic outcome ($i_t$) for the Aggregator.
\subsection{Characterization of $p(\cdot)$} To formalize the characterization, we split the V2G operation into two parts: $(1)$ \textit{smart charging}, and $(2)$ \textit{value estimator}. The \textit{smart charging} problem, parameterized in $a^\star(x_t)$, provides a stochastic charging profile, which is passed through a deterministic \textit{value estimator} to obtain the outcome and value for the aggregator. 

\textit{Smart charging:} Consider the decision vector $z=[P^{\text{c}},P^{\text{d}},\delta,E]^\top$ where $P^{\text{c}}\in \mathbb{R}^N$ is the charging power, $P^{\text{d}} \in \mathbb{R}^N$ is the discharging power for $N$ intervals, $E \in \mathbb{R}^N$ is the EV battery energy profile, and $\delta_i \in \{0,1\}, \forall i \in [N]$ refers the charging ($\delta_i=1$) or the discharging ($\delta_i=0$) mode of the EV. Given the user action $a^\star(x_t)$, the smart charging problem is compactly given as 
\begin{equation}\label{short_form}
    \max_{z} J(z,\zeta) \quad \text{s.t.} \quad z \in \mathcal{Z}(a^\star(x_t),\zeta),
\end{equation}
where $J$ encodes net charging benefit for the EV user, $\mathcal{Z}$ encodes the operational constraints, and the uncertain vector is $\zeta=[N,\eta^{\text{c}},\eta^{\text{d}},E_{\text{initial}},E_{\text{desired}}]$. A detailed formulation of smart charging is given in Appendix~\ref{sc_detailed}, along with an explanation of each parameter in $\zeta$. As the uncertain vector $\zeta$ is chosen by the user, we impose the following assumptions. 

\begin{assumption}[Smart charging setting] \label{ass:smart-charging}
Consider problem~\eqref{short_form} where the uncertain vector $\zeta$ is defined on the probability space $(\Omega, \mathcal{F}, \mathbb{P})$ with support $\Xi$. 
\begin{enumerate}
\renewcommand{\labelenumi}{(\roman{enumi})}
 \item {\bf Measurability:} For every $a\in\mathcal{A}$, the correspondence $\zeta\mapsto\mathcal{Z}(a,\zeta)$ has a measurable graph, while the map $(z,\zeta)\mapsto J(z,\zeta)$ is Borel measurable in $\zeta$. 

 \item {\bf Feasibility:} We assume that the feasibility set $\mathcal{Z}(a,\zeta)$ in \eqref{short_form} is non-empty and compact for all $\zeta \in \Xi$ and $a\in \mathcal{A}$. 

\end{enumerate}
\end{assumption}
Assumption~\ref{ass:smart-charging} helps construct a well-defined production function $p(\cdot)$, as we will show it later. The smart charging profile $\{\hat{P}^{\text{c}}(a^\star(x_t),\zeta),\hat{P}^{\text{d}}(a^\star(x_t),\zeta)\}$ is used in the next step for the value estimator (see Figure~\ref{v2g_schematic}) where $\hat{P}^{\text{c}}(a^\star(x_t),\zeta)$ and $\hat{P}^{\text{d}}(a^\star(x_t),\zeta)$ are part of the solution of \eqref{short_form} for a realization of the uncertain parameter $\zeta$. 

\textit{Value estimator:} To formalize the value estimator, we first define the net profit of the Aggregator (after transactions with DSO) stands as $J^P(a^\star(x_t),\zeta):= \Big[J(z^\star(a^\star(x_t)),\zeta) + (\hat{P}^d(a^\star(x_t),\zeta)-\hat{P}^c(a^\star(x_t),\zeta))^{\top}\lambda^g \Delta t \Big], $ where $z^\star(a^\star(x_t))$ is the solution of~\eqref{short_form}. Then we define the value estimator function $V: \mathcal{A}\times\Xi \rightarrow \{v^l,v^h\}$ such that, $v^\text{h} \geq v^\text{l}$, and 
\begin{equation} \label{val_est}
    V(a^\star(x_t),\zeta)= \begin{cases}
        & v^\mathrm{h} \quad \text{if,} \:J^P_t (a^\star(x_t),\zeta) \geq J', \\
        &v^\mathrm{l} \quad \text{otherwise,}
    \end{cases}
\end{equation}
where the threshold profit $J'$ and the values $v=[v^\mathrm{l},v^\mathrm{h}]^\top$ are chosen by the aggregator. Since $\zeta$ is a well-defined random variable, we assume that $v^{\mathrm{h}}$ and $v^\mathrm{l}$ are realized with probabilities $p^\mathrm{h}$ and $p^\mathrm{l}$, respectively. Essentially, this translates to defining the production function $p(a):= [p^\mathrm{l},p^\mathrm{h}]^\top$. To validate the mathematical correctness of $p(a)$, we need to show that it is possible to define $p^\mathrm{h}= \mathbb{P}_\zeta[J^P_t (a,\zeta) \geq J']$ and $p^\mathrm{l}=1-p^\mathrm{h}$ for any $a \in \mathcal{A}$. Thus, we provide the following Proposition. 
\begin{proposition} [Equivalence to principal-agent] \label{prop_1}
    Under Assumption~\ref{ass:smart-charging}, the proposed V2G incentive scheme constitutes a Principal-Agent problem wherein the smart charging formulation in \eqref{short_form} and the value estimator in \eqref{val_est} induce a well-defined production function $p:\mathcal{A}\rightarrow\Delta_2$. 
\end{proposition}
The proof is given in Appendix~\ref{proof_proposition}. Finally, the Aggregator aims to find the optimal contract $x^\star= \arg\max_{x\in[0,1]^2}\{U^{\mathrm{PA}}(x)= \sum_{i\in[2]}(v-x)_ip_i(a^\star(x))\}$ using the reward values (i.e., $y_t= (v-x_t)_i,\;i\sim p(a^\star(x_t))$) from the interaction with sequentially arriving EV users. To apply our proposed algorithm \texttt{Heteroscedastic GP-UCB}, we consider the stochastic modification in the EV user's utility (i.e., from $U^{\mathrm{A}}$ to $\tilde{U}^{\mathrm{A}}$), in addition to the required Assumptions for the algorithm.
\section{Numerical Experiments}
\label{sec:experiments}
We present two numerical experiments evaluating the performance of our algorithm against state-of-the-art benchmarks. First, we simulate a complex variant of the `High-low' example discussed in Section~\ref{sub_sec_discont.}, where the agent's action space is expanded to $n=3$. Second, we validate the proposed V2G incentive scheme through a realistic case study. We benchmark our performance against \texttt{Agnostic Zooming} and \texttt{Discover and Cover} algorithms.
\begin{table}
\centering
\caption{Performance comparison of different methods.
Effort levels correspond to $a_1$ (low), $a_2$ (medium), and $a_3$ (high).}
\label{high_low_table}

{\setlength{\tabcolsep}{2pt} 
\begin{tabular}{lccc}
\hline
\textbf{Agent} & Low & Med. & High \\
\hline
Cost $c(a_i)$
& \makecell{$\mathcal{N}$$(0,\,0.05)$}
& \makecell{$\mathcal{N}$$(0.1,\,0.05)$}
& \makecell{$\mathcal{N}$$(0.2,\,0.05)$} \\

Production $p(a_i)$
& \makecell{$[0.95,$$0.05]$}
& \makecell{$[0.2,$$0.8]$}
& \makecell{$[0.05,$$0.95]$} \\
\hline
\textbf{Principal} & \multicolumn{3}{c}{} \\
\hline
Outcome $i$ & Low ($i=0$) & \multicolumn{2}{c}{High ($i=1$)} \\
Value $v$ & $v_0=0.1$ & \multicolumn{2}{c}{$v_1=0.8$} \\
\hline
\end{tabular}
}
\end{table}
\subsection{Performance comparison in `High-low' example}
We retain the contract dimension $m=2$ while increasing the number of available actions to $n=3$. The detailed parameters of the feedback environment are provided in Table~\ref{high_low_table}. We compare the algorithms over a horizon of $T=1000$ rounds, averaging over $5$ independent episodes per algorithm. Since the exact optimal contract is intractable to compute \textit{a priori} in complex physical environments, we follow the convention established by \cite{ho2014adaptive} and evaluate empirical performance using the \textit{average utility reward}, defined as $U_T = \frac{1}{T}\sum_{t=1}^T \Tilde{U}^{\mathrm{PA}}(x_t)$. We explicitly note that this metric must be interpreted differently from cumulative regret. While cumulative regret ideally flattens to indicate sub-linear growth, a successful average utility curve converges upward to a positive horizontal asymptote, which represents the maximum achievable per-round yield. Consequently, an algorithm whose average utility converges to a higher, stable plateau demonstrates superior identification and exploitation of the optimal contract. 

As illustrated in Figure~\ref{regret_fig}, our proposed algorithm outperforms the benchmarks. Additionally, we conduct an ablation study to highlight the impact of heteroscedastic noise modeling and kernel selection in Appendix~\ref{ablation_appendix}. Overall, the observed performance gap validates the efficacy of our two key contributions: (1) optimization within a continuous domain, and (2) utility geometry-aware kernel choice paired with heteroscedastic noise modeling.
\begin{figure}[t]
    \centering

    \begin{subfigure}[t]{0.48\columnwidth}
        \centering
        \includegraphics[width=\linewidth]{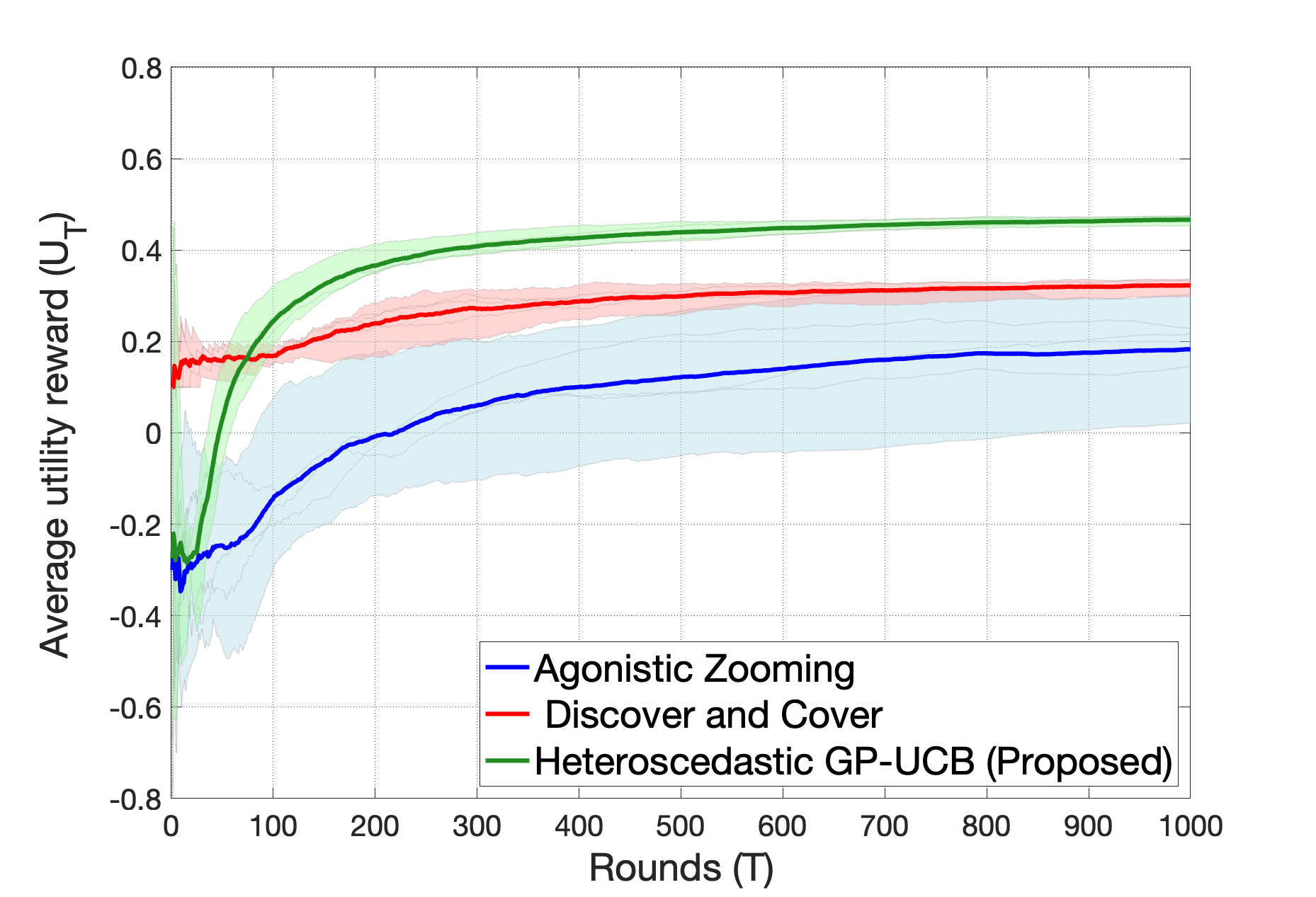}
        \caption{}
        \label{regret_fig}
    \end{subfigure}
    \hfill
    \begin{subfigure}[t]{0.48\columnwidth}
        \centering
        \includegraphics[width=\linewidth]{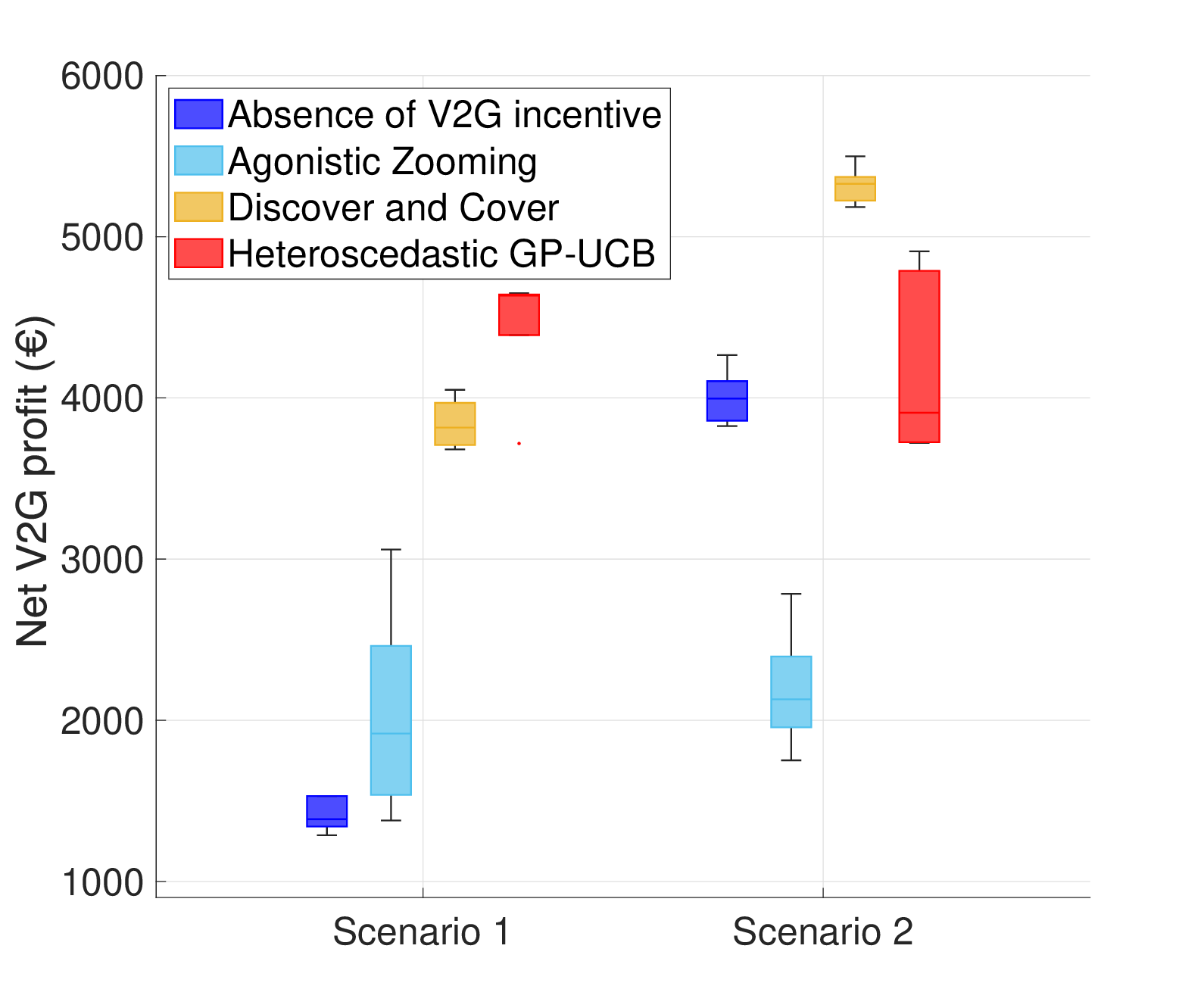}
        \caption{}
        \label{v2g_profit_fig}
    \end{subfigure}

    \caption{
    Performance comparison on the synthetic benchmark and V2G profit analysis:
    (a) Performance comparison of \texttt{Heteroscedastic GP-UCB} on the synthetic benchmark. The y-axis displays the \textbf{Average Utility Reward} ($U_T$), not cumulative regret. Therefore, the convergence of our proposed algorithm (green line) to a high, positive horizontal asymptote ($\approx 0.45$) indicates successful exploitation of the optimal contract, outperforming the baselines, which fail to converge to this optimal yield.
    (b) Comparison of the net V2G profit ($J_{\text{net}}$) for two scenarios, highlighting the impact of V2G incentive schemes under the benchmarked algorithms.
    }
    \label{fig:combined_results}
\end{figure}
\subsection{Results on V2G incentive design}
We evaluate the economic impact of our proposed incentive scheme on the Aggregator's net profit. To evaluate the aggregator's net profit from V2G participation of consecutive $T$ users, we define net V2G profit as $J_{\mathrm{net}}:= \vartheta\sum_{t=1}^T \tilde{U}^{\mathrm{PA}}(x_t),$ where $\vartheta$ is the re-normalization constant for conversion to Euro (\euro) currency. Specifically, we compare the net profit under two conditions: (i) a baseline scenario without an incentive mechanism, and (ii) a scenario where the V2G incentive mechanism is managed by the benchmarked algorithms. 
\textit{Experimental Setup:}
We construct a realistic simulation setup for this problem; detailed implementation specifics are provided in Appendix~\ref{implement_details}. We simulate two distinct user types, $\theta_1$ and $\theta_2$. The $\theta_1$ user incurs a low internal cost for V2G participation and is willing to participate without incentives. Conversely, the $\theta_2$ user is averse to participation due to high internal costs. We investigate two scenarios:
\textbf{Scenario 1 (High Aversion)} We sample $\theta_2$ with probability $0.95$. This dominance of reluctance creates an ideal testing ground for the impact of incentive design. \textbf{Scenario 2 (Mixed Population):} We uniformly sample $\theta_1$ and $\theta_2$. This presents a challenging setting where generating profit exceeding the baseline (no-incentive) case is difficult. We conduct the interaction with $T=1000$ sequential users, and gather results over $5$ independent episodes.

\textit{Discussion of Results:}
Figure~\ref{v2g_profit_fig} illustrates box plots of the key outcomes. In \textbf{Scenario 1}, we observe the substantial impact of the incentive design, with the baseline benchmark \texttt{Agnostic Zooming} achieving a minimum $\mathbf{43 \%}$ increase in $J_{\text{net}}$. Notably, our algorithm \texttt{Heteroscedastic GP-UCB} outperforms all benchmarks by a significant margin. Even in the challenging \textbf{Scenario 2}, our algorithm yields a $\mathbf{5 \%}$ higher profit compared to the absence of any V2G incentive scheme. This case study demonstrates that employing specialized learning-based algorithms can unlock diverse, economically viable business models for V2G technology. 
\section{Limitations and Future Work}
Numerical results demonstrate that our algorithm outperforms existing benchmarks. Furthermore, the V2G application confirms its robustness in complex, real-world settings. Nevertheless, we acknowledge the following limitations that suggest clear directions for future research.

\textit{Computational Scalability.} The kernel matrix inversion in \eqref{post_mean_var} incurs a computational complexity of $\mathcal{O}(T^3)$. While feasible for our experimental horizon ($T=1000$), this limits long-term deployment. Future work could integrate sparse GP approximations or variational inference to reduce computational complexity, enhancing real-time applications. 

\textit{Finite Agent Action Space.} While we eliminate exponential dependence on action count $n$, our second contribution still relies on a finite action set (Assumption~\ref{first_assumption}). Extending this framework to continuous agent action spaces would further theoretically improve upon existing methods. 

\appendix
\section{Technical Proofs}
\subsection{Auxiliary Lemma } \label{app_analytic_lemma_proof}
We first present a preparatory lemma that will be used to establish the main results of this study.
\begin{lemma}[Analytic kernel]\label{analytic_lemma}
    Let $D = [0,1]^m$ be the compact domain in $\mathbb{R}^m$. The kernel in \eqref{kernel_eq} is an analytic function on the domain $D \times D$.
\end{lemma}
\begin{proof}
    We observe that the kernel $k(x,y)$ is a composition of some functions. The broad idea of our proof is that we will show each composition is an analytic function in the domain $D \times D$. 
\begin{enumerate}
\renewcommand{\labelenumi}{(\roman{enumi})}
    \item The dot-product kernel $k_{\text{dot}}(x,y)=\sigma_v\langle x,y\rangle+\sigma_b$ is a polynomial in $x$ and $y$. Therefore, it is analytic in the domain. 
    \item Say $g(x,y)= (1 + k_{\text{dot}}(x,x))(1+k_{\text{dot}}(y,y))$. We find $k_{\text{dot}}(x,x)= \sigma_b^2 + \sigma_v^2\langle x,x\rangle > \sigma_b^2 >0$ that implies $(1 + k_{\text{dot}}(x,x)) > 1$. Therefore, $g(x,y) >1$ is an analytic function as it is the product of two analytic functions, and $\sqrt{g}$ is also an analytic function for all $g(x,y)>0$. 
    \item Representing $z(x,y)= \frac{k_{\text{dot}}(x, y)}{\sqrt{(1 + k_{\text{dot}}(x, x)) \cdot (1 + k_{\text{dot}}(y, y))}}$, it is an analytic function as the ratio of two analytic functions ($k_{\text{dot}}$ and $\sqrt{g}$) remain analytic as long as the denominator is non-zero. 
    \item Finally, it is well known that $\arcsin (z)$ is analytic if $z \in (-1,1)$. Therefore, we just need to check the range of $z(x,y)$. 
\end{enumerate}
Assume $a= k_{\text{dot}}(x,y), b=k_{\text{dot}}(x,x)$ and $c=k_{\text{dot}}(y,y)$. One could show using Cauchy-Schwarz inequality, $a^2 \leq bc$, furthermore it is easy to see that $\frac{a^2}{(1+b)(1+c)} <1$. Hence $|z(x,y)| < 1$ implies $z \in [-h,h]$ with $h <1$. This completes the proof. 
\end{proof}
\subsection{Proof of Proposition~\ref{eigen_val_lemma}}
\label{eigen_proof}

    It is known from the spectral theory of integral operators that analytic kernels have eigenvalues $\lambda_s$ which decay exponentially \cite{little1984eigenvalues,pietsch1987eigenvalues}. Using Lemma~\ref{analytic_lemma}, we can confirm the same for the eigenvalues of our proposed kernel $k(x,y)$. 

Following the above argument, we bound the tail sum $B_k(T_*)$ as
\begin{align*}
    B_k(T_*) \leq \sum_{s=T_*+1}^{\infty} Ce^{-vs^{1/m}} \leq C \int_{T_*}^{\infty} e^{-vs^{1/m}}ds
\end{align*}
for some constants $C,$ and $v$. It can be shown with some calculation that the above integral transforms to
\begin{align*}
     B_k(T_*) \leq C' \Gamma (m,vT_*^{1/m}),
\end{align*}
where $\Gamma$ is the upper incomplete gamma function and $C'$ is some constant. Using the formulae $\Gamma(m,q)= (m-1)! e^{-q}\sum_{k=0}^{m-1} (q^k /k!)$ \cite{gradshteyn2014table}, we can show for larger $q= vT_*^{1/m}$, our bound behaves as
\begin{equation}
   B_k(T_*) \leq \mathcal{O}((vT_*^{1/m})^{m-1}e^{-vT_*^{1/m}}).
\end{equation}
One can see that the exponential decay of $e^{-vT_*^{1/m}}$ is much stronger than the polynomial growth of $(vT_*^{1/m})^{m-1}$ as $T_* \rightarrow \infty$. Therefore, for any small $\delta > 0$, we can find a constant $C_\delta$ such that $(vT_*^{1/m})^{m-1} \leq C_\delta e^{\delta (vT_*^{1/m})}$. Substituting this into the above equation
\begin{align*}
   B_k(T_*) &\leq \mathcal{O}\left( e^{\delta v T_*^{1/m}} \cdot e^{-vT_*^{1/m}} \right)
   = \mathcal{O}\left( e^{-(1-\delta)v T_*^{1/m}} \right).
\end{align*}
By choosing $\delta$ such that $\beta = (1-\delta)v > 0$, we obtain
\begin{align*}
   B_k(T_*) \leq \mathcal{O}(e^{-\beta T^{1/m}_*}).
\end{align*}
This concludes the proof. \hfill $\square$
\subsection{Proof of Proposition~\ref{prop_1}} \label{proof_proposition}
We first recall the setup of the principal-agent problem with respect to the Figure~\ref{v2g_schematic}. The principal (aggregator) shows the contract $x_t \in [0,1]^2$ to the agent (EV user), and the agent takes a strategic action $a^\star(x_t)$ in hindsight. As a result, an outcome is realized following the condition in \eqref{val_est}. The main focus here is to understand the existence of a production function $p(\cdot)$ as defined in \eqref{agent_bp}. If we manage to show that $J^P(a^\star(x_t),\zeta)= \left [J(z^\star(a^\star(x_t)),\zeta) + (\hat{P}^d(a^\star(x_t),\zeta)-\hat{P}^c(a^\star(x_t),\zeta))^{\top}\lambda^g \Delta t \right ], $ is a well-defined random variable in \eqref{val_est}, we can confirm the existence of a production function. To establish that $J^P(a^\star(x_t),\zeta)$ is a random variable, we proceed by the following steps. 
\begin{enumerate}
    \item Since \eqref{obj_sc} indicates that $J$ is continuous in $z$, we apply Berge's maximum theorem \cite[Theorem $17.31$]{aliprantis2006infinite} to confirm that $J(z^\star(a^\star(x_t)),\zeta)$ is continuous in $\zeta$.
    \item Moreover, since $J(z,\zeta)$ is a Borel measurable function of $\zeta$ (Assumption~\ref{ass:smart-charging}), both the optimizer $z^\star(a^\star(x_t))$ and optimal objective value $J(z^\star(a^\star(x_t),),\zeta)$ become Borel measurable in $\zeta$, where the feasibility of the optimizer $z^\star(a^\star(x_t))$ in~\eqref{short_form} is ensured by Assumption~\ref{ass:smart-charging}.
\end{enumerate}
Finally, $J^P(a^\star(x_t),\zeta)$ is a linear combination of measurable functions, which makes it a measurable function of $\zeta$. Hence, it is a well-defined random variable, which makes the definition $p^\mathrm{h}$ and $p^\mathrm{l}$ consistent. This concludes the proof. \hfill $\square$
\section{Detailed smart charging formulation }  \label{sc_detailed}
Consider $P^{\text{c}}\in \mathbb{R}^N$ is the charging profile, and  $P^{\text{d}} \in \mathbb{R}^N$ is the discharging profile for $N$ intervals. The aggregator defines the buying price $\lambda^b_i= c \lambda^{g}_i, \: \forall i \in [N]$, and the selling price $\lambda^s_i= c' \lambda^{g}_i, \: \forall i \in [N]$, where $c,c'$ are the aggregator-defined constants resposible for profit margin such that $c \in (0,1], c' \geq 1$. $E \in \mathbb{R}^N$ is the EV battery energy profile with $\overline{E}$ and $\underline{E}$ being the limits of the energy level. Furthermore, we use $\delta_i \in \{0,1\}, \forall i \in [N]$ to refer to the charging ($\delta_i=1$) or the discharging ($\delta_i=0$) mode of the EV, and we use $\eta^c/\eta^d$ to refer to the charging/discharging coefficients embedding the efficiency and battery capacity of the EV, respectively. Based on the active load profile up to time $t$, we also have the bounds on the charging/discharging power as $\overline{P}_t$ and $\underline{P}_t$, respectively. Given these terminologies, the smart charging algorithm solves the following optimization problem: 
\begin{subequations} \label{charging_prob}
    \begin{align}
        \max_{P^{\text{c}},P^{\text{d}},\delta,E} \quad & J:= (P^{\text{c}}(\lambda^{s})^{\top}-P^{\text{d}}(\lambda^{b})^{\top})\Delta t  \nonumber \\
        &- \alpha (\lVert P^c \rVert_2^2 + \lVert P^d \rVert_2^2 ) , \label{obj_sc} \\
        \text{s.t.} \quad \quad & E_{i+1}= E_i + \eta^c P^{\text{c}} \Delta t- \eta^d P^{\text{d}} \Delta t,  \forall i \in [N], \\
        & 0 \leq P^{\text{c}}_i \leq \delta_i \overline{P}_{t,i},\: \forall i \in [N], \\
        & 0 \leq P^{\text{d}}_i \leq (1-\delta_i) \underline{P}_{t,i},\: \forall i \in [N], \\
        &\underline{E} \leq E_i \leq \overline{E}, \: \forall i\in [N], \\
        & E_1= E_{\text{initial}}, \: E_{N+1}= E_{\text{desired}}, \\
        &\delta_j=1, \; \forall j \in [N], \; \mathrm{if}\; a^\star(x_t)=1,
    \end{align}
\end{subequations}
\begin{figure}
    \centering
    \includegraphics[width=0.6\linewidth]{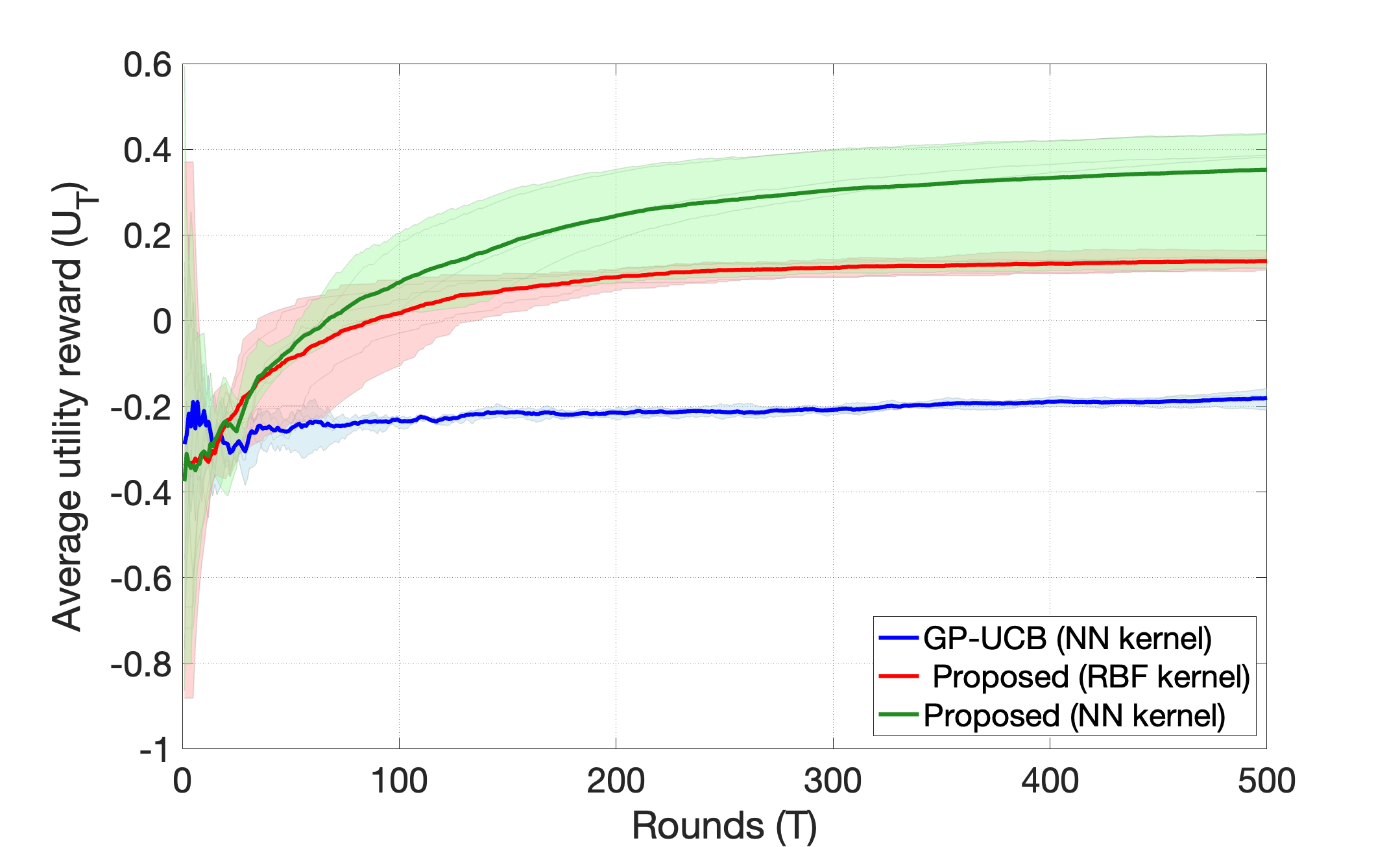}
    \caption{Assessing performance of \texttt{Heteroscedastic GP-UCB} (paired with Neural Network kernel) under two scenarios. In scenario $1$, it outperforms GP-UCB \cite{srinivas2012information} with a Neural Network kernel, validating the effect of considering heteroscedastic noise. In scenario $2$, it also outperforms the radial basis function (RBF) kernel, validating our kernel choice.}
    \label{ablation_fig}
\end{figure}
where the objective is to maximize the net profit $J \in \mathbb{R}$ for the EV user, the constraints include energy dynamics of the EV battery, limiting bounds for charging/discharging power and energy, $\alpha$ is a regularizer coefficient, $\Delta t$ is the time interval, and finally, $N, E_{\text{initial}}, E_{\text{desired}}$ are the parameters given by the EV users.
\section{Implementation details of numerical experiments and additional results} \label{implement_details}
\begin{table}[t]
\centering
\caption{Details of uncertain and fixed parameters for numerical simulation.}
\label{tab:parameters}

\begin{minipage}[t]{0.48\columnwidth}
\centering
\caption*{(a) Uncertain parameters}
\label{exp_data}
{\setlength{\tabcolsep}{4pt}
\begin{tabular}{lc}
\hline
\textbf{Uncertain parameters} & \textbf{Distribution} \\
\hline
$N$ & Uniform($20,24$) \\
$\eta^c$ and $\eta^d$ & Uniform(0.9,0.98) \\
$E_{\text{initial}}$ & Uniform(10,15) \\
$E_{\text{desired}}$ & Uniform(20, $N\overline{P}_t\eta^c\Delta t$) \\
\hline
\end{tabular}
}
\end{minipage}
\hfill
\begin{minipage}[t]{0.48\columnwidth}
\centering
\caption*{(b) Fixed parameters}
\label{exp_data_fixed}
{\setlength{\tabcolsep}{4pt}
\begin{tabular}{lc}
\hline
\textbf{Fixed parameters} & \textbf{Values} \\
\hline
$\underline{E}$ & $5$ kWh \\
$\overline{E}$ & $90$ kWh \\
$\overline{P}_t$ and $\underline{P}_t$ & $11$ kW \\
$c$ & $0.7$ \\
$c'$ & $1.2$ \\
$J'$ & $80$ \euro \\
$\Delta$ & $10$ \euro \\
$v^h$ & $0.9$ \\
$v^l$ & $0$ \\
$\Delta t$ & $15$ minutes \\
\hline
\end{tabular}
}
\end{minipage}

\end{table}
\subsection{Implementation details}
We have used Assumption~\ref{ass:smart-charging} to construct the random vector $\zeta$ in Table~\ref{exp_data} such that its realized values appear realistic for the application, while maintaining the feasibility of the smart charging problem in \eqref{charging_prob}. The values of the fixed parameters in \eqref{charging_prob} and \eqref{val_est} are given in Table~\ref{exp_data_fixed}. In addition, the tariff data from DSO ($\lambda^g$) is taken from Figure $7$ of \cite{lai2022pricing}. Given $\lambda^g$, one can easily construct $\lambda^b$ and $\lambda^s$ using the profit margin constants $c$ and $c'$ from Table~\ref{exp_data_fixed}.  
\subsection{Additional results} \label{ablation_appendix}
We present an ablation study in Figure~\ref{ablation_fig} to decouple the contributions of heteroscedastic variance modeling and kernel selection. First, to isolate the impact of the noise model, we compare our approach against standard homoscedastic GP-UCB \cite{srinivas2012information} equipped with the Neural Network kernel. Second, to validate the kernel geometry, we substitute the NN kernel with the standard Radial Basis Function (RBF) kernel within our framework. The results demonstrate that our proposed method outperforms both variations, confirming that both the heteroscedastic structure and the non-stationary kernel are essential for our problem setup. 
\bibliographystyle{plain}        
\bibliography{ref}
\end{document}